\documentclass[11pt]{article}

\usepackage{fullpage}
\usepackage{amsmath,amsfonts,amsthm,mathrsfs,xspace,hyperref,graphicx}
\usepackage{endnotes}
\usepackage{bm}
\usepackage{enumitem}
\usepackage{amssymb,latexsym}
\usepackage{relsize}
\usepackage{mathpazo}
\usepackage{cleveref}
\usepackage{stmaryrd}
\usepackage{xcolor}

\definecolor{weborange}{rgb}{.8,.3,.3}
\definecolor{webblue}{rgb}{0,0,.8}
\definecolor{internallinkcolor}{rgb}{0,.5,0}
\definecolor{externallinkcolor}{rgb}{0,0,.5}

\definecolor{DarkBlue}{rgb}{0,0,0.8}  \definecolor{DarkOrange}{rgb}{0.8,0.4,0}  \def\mylinkcolor{DarkBlue}

\hypersetup{colorlinks=true, urlcolor=\mylinkcolor,
  linkcolor=\mylinkcolor, citecolor=\mylinkcolor}

\newtheorem{theorem}{Theorem}[section]
\newtheorem{proposition}[theorem]{Proposition}
\newtheorem{lemma}[theorem]{Lemma}
\newtheorem{claim}[theorem]{Claim}

\newtheorem{corollary}[theorem]{Corollary}
\newtheorem{definition}[theorem]{Definition}
\theoremstyle{definition}

\DeclareMathOperator*{\Ex}{\mathbb{E}}

\newcommand{\ket}[1]{{|#1\rangle}}
\newcommand{\bra}[1]{{\langle#1|}}
\newcommand{\bigket}[1]{{\left |#1 \right \rangle}}
\newcommand{\bigbra}[1]{{\left \langle#1 \right|}}

\newcommand{\ketbra}[2]{|#1\rangle\! \langle #2|}
\newcommand{\Tr}{\mbox{\rm Tr}}
\newcommand{\Id}{\mathbb{I}}

\newcommand{\C}{\ensuremath{\mathbb{C}}}
\newcommand{\N}{\ensuremath{\mathbb{N}}}

\newcommand{\E}{\mathcal{E}}

\newcommand{\mH}{\mathcal{H}}
\newcommand{\X}{\mathcal{X}}
\newcommand{\Y}{\mathcal{Y}}
\newcommand{\A}{\mathcal{A}}
\newcommand{\B}{\mathcal{B}}

\newcommand{\D}{\mathrm{D}}
\newcommand{\I}{\mathrm{I}}
\newcommand{\F}{\mathrm{F}}

\newcommand{\puretomixed}[1]{\llbracket #1 \rrbracket}

\newcommand{\eps}{\varepsilon}
\newcommand{\poly}{\mathrm{poly}}

\newcommand{\wt}[1]{\widetilde{#1}}
\newcommand{\what}[1]{\hat{#1}}
\newcommand{\eval}{\omega^*}
\newcommand{\val}{\omega}

\newcommand{\strategy}{\mathscr{S}}

\renewcommand{\P}{\mathsf{P}}

\newcommand{\ba}{\mathbf{a}}
\newcommand{\bb}{\mathbf{b}}

\newcommand{\bx}{\mathbf{x}}
\newcommand{\by}{\mathbf{y}}

\newcommand{\br}{\mathbf{r}}
\newcommand{\bX}{\mathbf{X}}
\newcommand{\bY}{\mathbf{Y}}

\newcommand{\bR}{\mathbf{R}}
\newcommand{\bA}{\mathbf{A}}
\newcommand{\bB}{\mathbf{B}}
\newcommand{\bOmega}{{\boldsymbol{\Omega}}}
\newcommand{\bomega}{{\boldsymbol{\omega}}}
\newcommand{\bD}{\mathbf{D}}
\newcommand{\bM}{\mathbf{M}}
\newcommand{\bU}{\mathbf{U}}

\newcommand{\bV}{\mathbf{V}}

\newcommand{\mi}{{-i}}

\newcommand{\dummy}{{\mathsmaller{\perp}}}

\newcommand\dummyx{{\dummy\!/x}}

\newcommand{\ac}{\mathrm{Z}}

\newcommand{\xp}{{x,\dummy}}
\newcommand{\py}{{\dummy, y}}
\newcommand{\xy}{{x,y}}
\newcommand{\pxy}{{\dummy\!/x,y}}
\newcommand{\pxp}{{\dummy\!/x,\dummy}}

\newcommand{\pp}{{\dummy, \dummy}}

\newcommand{\ssxp}{{\br_\mi , x, \dummy}}
\newcommand{\sspy}{{\br_\mi , \dummy,y}}
\newcommand{\sspp}{{\br_\mi , \dummy,\dummy}}
\newcommand{\ssxy}{{\br_\mi , x,y}}
\newcommand{\sspxy}{{\br_\mi , \dummyx,y}}
\newcommand{\sspxp}{{\br_\mi , \dummyx,\dummy}}

\newcommand{\Qsf}{\mathsf{Q}}
\newcommand{\Ssf}{\mathsf{S}}
\newcommand{\Rsf}{\mathsf{R}}

\begin{document}

\title{Anchored parallel repetition for nonlocal games}

\author{Mohammad Bavarian\thanks{Work completed as a graduate student at MIT. Email: \texttt{mobavarian@gmail.com}.} \and
Thomas Vidick~\thanks{Department of Computing and Mathematical Sciences, California Institute of Technology, USA. Email: \texttt{vidick@caltech.edu}.} \and
Henry Yuen~\thanks{Department of Computer Science, Columbia University, USA. Email: \texttt{henry.yuen@columbia.edu}.}}

\date{}
\maketitle

\begin{abstract}
We introduce a simple transformation on two-player nonlocal games, called ``anchoring'', and prove an exponential-decay parallel repetition theorem for all anchored games in the setting of quantum entangled players. This transformation is inspired in part by the Feige-Kilian transformation (SICOMP 2000), and has the property that if the quantum value of the original game $G$ is $v$ then the quantum value of the anchored game $G_\dummy$ is $1 - (1 - \alpha)^2 \cdot (1 - v)$ where $\alpha$ is a parameter of the transformation. In particular the anchored game has quantum value $1$ if and only if the original game $G$ has quantum value $1$. This provides the first gap amplification technique for general two-player nonlocal games that achieves exponential decay of the quantum value. 

\end{abstract}

\section{Introduction}\label{sec:anchorpr_intro}

A \emph{nonlocal game} is specified by finite question sets $\X, \Y$, finite answer sets $\A, \B$, a probability distribution $\mu$ over $\X\times\Y$, and a verification predicate $V : \X \times \Y \times \A \times \B\rightarrow \{0,1\}$ that determines valid question and answer tuples. The game is played as follows. A referee samples questions $(x,y) \in \X \times \Y$ according to $\mu$ and sends $x$ to the first player and $y$ to the second. Each player replies with an answer, $a\in \A$ and $b \in \B$ respectively. The referee accepts if and only if $V(x,y,a,b) = 1$, in which case we say that the players win the game. Nonlocal games have been studied in a variety of settings, ranging from hardness of approximation~\cite{khot2002power} (in this literature they are called \emph{two-prover one-round games}), interactive proof systems~\cite{ben1988,fortnow1988}, and the study of Bell inequalities and non-locality in quantum physics~\cite{bell1964,chsh_paper,cleve2004consequences}. Recently, nonlocal games has become central objects of study in quantum complexity theory and quantum information, especially given their role in the study of quantum interactive proofs~\cite{reichardt2013classical,ji2020mip}, quantum cryptography~\cite{vazirani2012certifiable,vazirani2014fully}, and fundamental questions about quantum entanglement~\cite{coladangelo2018unconditional,slofstra2019set}.

The main quantity associated with a nonlocal game $G$ (which we'll often simply refer to as a \emph{game}) is its \emph{value}, which is the maximum acceptance probability achievable by the players, where the probability is taken over the questions as chosen by the referee and the players' answers. Different notions of value arise from different restrictions on allowed strategies for the players.  The most relevant for us are the \emph{classical value} denoted by $\val(G)$ and the \emph{quantum value} denoted by $\eval(G)$. The former is obtained by restricting the players to classical strategies, where each player's answer is a function of its question only (both private and shared randomness are in principle allowed, but easily seen not to help). The latter allows for quantum strategies, in which each player's answer is obtained as the outcome of a local measurement performed on a (finite-dimemsional) quantum state shared by the players. While the use of quantum states {does not} allow communication between the players, as first conclusively demonstrated by Bell~\cite{bell1964} it does allow for correlations between their questions and answers that cannot be reproduced by any classical strategy.

We study the behavior of the quantum value of games under parallel repetition. 
In the $n$-fold parallel repetition $G^n$ of a game $G$ the referee samples $(x_1,y_1),\ldots,(x_n,y_n)$ independently from $\mu$ and sends $(x_1,\ldots,x_n)$ to the first player and $(y_1,\ldots,y_n)$ to the second. The players respond with answer tuples $(a_1,\ldots,a_n)$ and $(b_1,\ldots,b_n)$ respectively. The players win if and only if their answers satisfy $V(x_i,y_i,a_i,b_i) = 1$ for all $i$. 

Clearly, if the players play each instance of $G$ in $G^n$ independently of each other then their success probability is the $n$-th power of their success probability in $G$. Informally, a parallel repetition theorem for a class of games and a class of strategies shows that no strategy can avoid this exponential scaling of the success probability with the number of repetitions. The main obstacle to proving a parallel repetition theorem is that players need not play each instance independently. For example, their answers for the $i$-th instance of $G$ may depend on their questions in the $j$-th instance for $j \neq i$. Indeed, for classical (as well as for quantum) strategies it is known that there are games $G$ for which non-product strategies enable the players to win $G^n$ with probability greater than $\val(G)^n$~\cite{feige2002error, raz2011counterexample,chung2015parallel}.

For the case of classical strategies the parallel repetition theorem of Raz~\cite{raz1998parallel} establishes that if $G$ is a game such that $\val(G)<1$ the value $\val(G^n)$ decays exponentially with $n$. The two following decades have seen a substantial amount of research on this question, connecting the problem of parallel repetition to topics such as the Unique Games conjecture, hardness of approximation, communication complexity, and more~\cite{barak2008rounding, rao2011parallel,braverman2013direct}. The most important application of parallel repetition is its use as a generic and efficient method for performing \emph{gap amplification}, or \emph{hardness amplification}. Suppose a certain problem -- deciding membership in a language $L$, or breaking a given cryptosystem -- has been reduced to the task of distinguishing between $\val(H) = 1$ and $\val(H) < \delta$ for a certain $H$. Parallel repetition can be employed to argue that, by taking $H = G^n$ for some game $G$, this task is as hard as that of distinguishing between $\val(G) = 1$ and $\val(G) < 1-\eps$ when $n \geq \poly(\eps^{-1},\log \delta^{-1})$.

In recent years there has been much interest in obtaining hardness amplification results for games with quantum entangled players --- in particular, obtaining an analogue of Raz's theorem for the quantum value of nonlocal games. However, it has been a challenge to extend the techniques used to prove classical parallel repetition theorems to the setting of quantum strategies; one significant difficulty is that there is no \emph{a priori} upper bound on the amount of entanglement needed to play a given game optimally. Thus an important motivation for studying this question, aside from its application to hardness amplification, is to further develop mathematical tools for analyzing quantum entanglement in interactive protocols. Despite much research---and partial results, as surveyed in Section~\ref{sec:related}---it remains an open question as to whether an analogue of Raz's theorem holds for entangled games. In this paper, we make progress on this question.

\paragraph{Main result.}
Our main result is that there exists a polynomial-time transformation, called \emph{anchoring}, that takes as input a parameter $0 < \alpha < 1$ and the description of a nonlocal game $G$ and returns the description of a game $G_\dummy$ such that the quantum value of $G_\dummy$ decays exponentially under parallel repetition. We first define the anchoring transformation.

\begin{definition}[The anchoring transformation] \label{def:anchor_basic} 
Let $G$ be a nonlocal game with question distribution $\mu$ on $\X\times \Y$ and verification predicate $V$ and let $0< \alpha\leq 1$. In the \emph{$\alpha$-anchored game $G_{\dummy}$} the referee first chooses a question pair $(x,y)\in \X\times \Y$ according to $\mu$. Independently and with probability $\alpha$ the referee replaces each of $x$ and $y$ with an auxiliary ``anchor'' symbol $\dummy$ to obtain the pair $(x', y')\in (\X\cup \{\dummy \}) \times (\Y\cup \{\dummy \})$ which is sent to the players as their respective questions. If any of $x',y'$ is $\dummy$ the referee accepts regardless of the players' answers; otherwise, the referee checks the players' answers according to the predicate $V$.
\end{definition}

The following simple lemma relates the quantum values of $G$ and $G_\dummy$.

\begin{lemma}
Let $0< \alpha < 1$ and let $G$ be a nonlocal game. Then 
\begin{gather*}
	\omega^*(G_\dummy) \,=\, 1 - (1 - \alpha)^2 \cdot (1 - \omega^*(G))\;.
\end{gather*}
In particular, the quantum value of $G_\dummy$ is $1$ if and only if the quantum value of $G$ is $1$. (Moreover, the same equation holds for the classical value.)
\end{lemma}
\begin{proof}
Let $\X,\Y$ denote the question alphabets for the game $G$. 
	We first show that $\omega^*(G_\dummy) \geq 1 - (1 - \alpha)^2 \cdot (1 - \omega^*(G))$. Let $\strategy$ be a quantum strategy for $G$ that succeeds with probability $p$. Define a strategy $\strategy_\dummy$ for $G_\dummy$ as follows: upon receipt of a question $x\neq\dummy$ the player generates an answer according to $\strategy$. Otherwise, if its question is $\dummy$ then the player outputs a fixed answer from $\cal{A} \cup \cal{B}$. The success probability of this strategy in $G_\dummy$ is
	\[
		(1 - (1-\alpha)^2) + (1 - \alpha)^2 \cdot p \,=\, 1 - (1 - \alpha)^2 \cdot (1 - p)\;,
	\]
	where $(1 - \alpha)^2$ is the probability that neither player receives the question $\dummy$. Taking the supremum over strategies $\strategy$ yields the desired lower bound on $\eval(G_\dummy)$.
	
	On the other hand, it also is the case that $\omega^*(G_\dummy) \leq 1 - (1 - \alpha)^2 \cdot (1 - \omega^*(G))$. Let $\strategy_\dummy$ be a strategy for $G_\dummy$ that succeeds with probability $q$. Note that by ignoring the measurements that correspond to the question $\dummy$, $\strategy_\dummy$ naturally implies a strategy $\strategy$ for $G$. Since in $G_\dummy$ the players automatically win whenever either of them receive the question $\dummy$ their success probability can be expressed as $q = (1 - (1-\alpha)^2) + (1 - \alpha)^2 \cdot p$ where $p$ denotes the probability that the players succeed in $G$ using the strategy $\strategy$. Taking the supremum over strategies $\strategy_\dummy$ yields the desired lower bound on $\eval(G)$.
	
The same proof yields the same relation between the classical values $\val(G_\dummy)$ and $\val(G)$.
\end{proof}

To give an example: for $\alpha = 1 - \frac{\sqrt{3}}{2}$ it holds that $\eval(G_\dummy) =\frac{3}{4} \eval(G) + \frac{1}{4}$. One can think of $G_\dummy$ as playing the original game $G$ with probability $3/4$, and a trivial game with probability $1/4$. The term ``anchored'' refers to the fact that question pairs chosen according to $\mu$ are all  ``anchored'' by a common question $(\dummy,\dummy)$. Though the presence of the anchor question makes the game $G_\dummy$ {easier} to play than the game $G$ (in the sense that $\eval(G_\dummy) \geq \eval(G)$)  it facilitates showing that the quantum value of the repeated game $G_\dummy^n$ decays exponentially with $n$, as established by the following.

\begin{theorem}\label{thm:main-informal}
Let $0 < \alpha < 1$.
Let $G_\dummy$ be an $\alpha$-anchored game with answer alphabets $\A,\B$ satisfying $\eval(G_\dummy) < 1-\eps$. Then for all $n\geq 1$,
	$$ \eval(G^n_\dummy)\leq \frac{4}{\eps} \, \exp \Big ( - \frac{c \cdot \alpha^{48} \cdot \eps^{17} \cdot n}{s} \Big)\;,$$
	where $s = \max \{ \log |\A \times \B|, 1 \}$ and $c > 0$ is a universal constant.
\end{theorem}

We prove \Cref{thm:main-informal} as a corollary to a more general result that relates the \emph{minimum dimension} of any quantum strategy for $G_\dummy^n$ that succeeds with a sufficiently large probability in $G_\dummy^n$ to the minimum dimension required to succeed in $G_\dummy$ with probability at least $1 - \eps$. This more general result is presented as \Cref{thm:anchorpr_quantum}.

The idea of modifying  a game to facilitate its analysis under parallel repetition originates in the work of Feige and Kilian~\cite{feige2000two}, which predates Raz's parallel repetition theorem. 
Feige and Kilian introduce a transformation that converts an arbitrary game $G$ to a so-called \emph{miss-match} game $G^{FK}$. The transformation is {value-preserving} in the sense that there is a precise affine relationship $\val(G^{FK}) = (2 + \val(G))/3$. Furthermore Feige and Kilian show that the value of the $n$-fold repetition of $G^{FK}$ decays {polynomially} in $n$ whenever $\val(G)<1$. This enables them to establish a general gap amplification result without having to prove a parallel repetition theorem for arbitrary games. This is sufficient for many applications, including to hardness of approximation, for which it is enough that the gap amplification procedure be efficient and value-preserving. 

We adopt a similar approach to that of Feige and Kilian by providing an arguably even simpler transformation which \emph{preserves both the classical and quantum value} of a game and for which we are able to prove an exponential decay under parallel repetition. In contrast the transformation considered by Feige and Kilian does not in general preserve the quantum value: there are examples of $G$ such that $\eval(G)=1$ but $\eval(G^{FK})<1$.\footnote{Such an example can be constructed from the Magic Square game $G_{MS}$~\cite{mermin1990simple,aravind2002simple}, which satisfies $\eval(G_{MS})=1$, and using the techniques in~\cite{ito2009oracularization} to show that $\eval(G_{MS}^{FK})<1$; we omit the details.}

 \subsection{Related work}\label{sec:related}

We refer to the surveys by Feige and Raz~\cite{feige1995,raz2010parallel} for an extensive historical account of the classical parallel repetition theorem and its connections to the hardness of approximation and multiprover interactive proof systems, and instead focus on more recent results, specifically those pertaining to quantum parallel repetition.

The first result on the parallel repetition of entangled-player games was obtained by Cleve, Slofstra, Unger, and Uphadyay~\cite{cleve08} for the class of XOR games. This was extended to the case of unique games by Kempe, Regev and Toner \cite{kempe2008}. Kempe and Vidick~\cite{kempe2011parallel} studied a Feige-Kilian type repetition for the quantum value of nonlocal games and obtained a polynomial rate of decay for this type of repetition. The Feige-Kilian transformation does not in general preserve the quantum value, and their result does not provide a hardness amplification technique for arbitrary nonlocal games.

Dinur, Steurer and Vidick~\cite{DinurSV14} extended the analytical framework of Dinur and Steurer~\cite{dinur2014analytical} to obtain an exponential-decay parallel repetition theorem for the quantum value of the class of projection games. Chailloux and Scarpa~\cite{chailloux2014parallel} and  Jain, Pereszl\'{e}nyi, and Yao~\cite{jain2014parallel} prove exponential-decay parallel repetition for \emph{free games}, i.e. games where the questions are sampled independently. Their analysis, as well as the follow-up work Chung, Wu, and Yuen~\cite{chung2015parallel}, is based on extending the information-theoretic techniques used by Raz~\cite{raz1998parallel} and Holenstein~\cite{Hol09} to study the parallel repetition of the classical value of games.

\paragraph{Subsequent work.} Since the original posting of this work several related papers have appeared. In a separate paper~\cite{bavarian2017parallel} we analyze a different hardness amplification method called ``fortification'',  first introduced by Moshkovitz~\cite{moshkovitz2014parallel} in the context of classical parallel repetition and obtain exponential-decay parallel repetition bounds for the quantum value of fortified games. Later the second author~\cite{yuen2016parallel} showed that for \emph{all} nonlocal games $G$, the quantum value of $G^n$ must decay to $0$ at a polynomially-fast rate (provided that $G$ has quantum value less than one). Finally, in~\cite{jain2020direct} Jain and Kundu give an alternate proof of \Cref{thm:main-informal} and use the anchoring transformation to study {direct product theorems} for one-way quantum communication complexity.

\paragraph{Comparison with earlier versions of this paper.} Earlier versions of this paper~\cite{bavarian2015anchoring,bavarian2017hardness} include a parallel repetition result for games with more than two players. (For the case of more than two players it is not known if even the classical value decreases exponentially under standard repetition.) These extensions are omitted from the current version because they are subsumed by a subsequent paper of Dinur, Harsha, Venkat and Yuen~\cite{dinur2017multiplayer}; this allows us to focus on the main contribution of the paper, the analysis of anchoring for two-player games with quantum players.

\paragraph{Acknowledgments.} T.V. is supported by NSF CAREER Grant CCF-1553477, AFOSR YIP award number FA9550-16-1-0495, MURI Grant FA9550-18-1-0161 and the IQIM, an NSF Physics Frontiers Center (NSF Grant PHY-1125565) with support of the Gordon and Betty Moore Foundation (GBMF-12500028). H.Y. is supported by an NSERC Discovery Grant, a Google Research Award, and AFOSR award FA9550-21-1-0040.

\section{Technical overview}
\label{sec:anchorpr_technical}

Essentially all known proofs of parallel repetition proceed by reduction. One shows how a strategy with a sufficiently large value in the repeated game $G^n$ can be ``rounded'' into a strategy $\strategy$ for $G$ with value strictly greater than $\val(G)$ (or $\eval(G)$ in the quantum case), yielding a contradiction. 

Let $\strategy^n$ be a strategy for $G^n$ such that $\eval(G^n) \gg \eval(G)^n$, and let us aim to define a good strategy $\strategy$ for $G$. (The notation $\strategy^n$ does \emph{not} denote that $\strategy^n$ is a product strategy.) By a straightforward inductive argument one can show that there must exist a set $C \subset [n]$ and an index $i \in [n] \setminus C$ such that $\P(W_i | W_C) > \eval(G) + \delta$, where $W_i$ is the event that the players win the $i$-th instance of $G$ in $G^n$ and $W_C$ is the event that the players win all instances indexed by $C$. Given a pair of questions $(x,y)$ in $G$ the strategy $\strategy$ embeds them in the $i$-th coordinate of a $n$-tuple of questions $(\bx,\by)$ with
\begin{gather*}
\bx = (\bx_1, \bx_2, \ldots, \bx_{i-1}, x, \bx_{i+1}, \ldots, \bx_n) \;,\\
\by = (\by_1, \by_2, \ldots, \by_{i-1}, y, \by_{i+1}, \ldots, \by_n)\;,
\end{gather*}
that is approximately distributed according to the distribution of the players' question tuples in the game $G^n$ when they use the strategy $\strategy^n$, conditioned on the $i$-th question pair being $(x,y)$ and conditioned on the event $W_C$. We denote such conditional distributions as $\P_{\bX \bY | \bX_i = x, \bY_i = y,W_C}$, where bold variables are used to denote tuples. The players then execute $\strategy^n$ on  $\bx$ and $\by$ respectively to obtain answers $\ba = (\ba_1,\ldots,\ba_n)$ and $\bb = (\bb_1,\ldots,\bb_n)$. Finally, they return $(\ba_i,\bb_i)$ as their answers in $G$. 

The answers $(\ba_i,\bb_i)$ conditioned on question pair $(x,y)$ generated using the strategy $\strategy$ is $\delta'$-close to the probability distribution $\P_{\bA_i \bB_i | \bX_i = x,\bY_i = y,W_C}$, for some $\delta' \ll \delta$. It can be shown that for most indices $i \in [n] \setminus C$, the conditional distribution $\P_{\bX_i \bY_i | W_C}$ is also $\delta'$-close to the question distribution $\mu$ of $G$; assume that such an index $i$ was selected. Then the distribution of question and answers $(x,y,\ba_i,\bb_i)$ using this strategy is going to be $2\delta'$-close to the distribution $\P_{\bX_i \bY_i \bA_i \bB_i | W_C}$, which is the distribution of questions and answers in the $i$-th coordinate conditioned on the event $W_C$. Thus the probability that $V(x,y,\ba_i,\bb_i) = 1$ is going to be $2\delta'$-close to $\P(W_i | W_C)$, which provided $\delta'$ is small enough compared to $\delta$ is strictly greater than $\eval(G)$. This yields the desired contradiction.

Since $\strategy^n$ is not necessarily a product strategy, conditioning on $W_C$ may introduce correlations that make $\P_{\bX \bY | \bX_i = x, \bY_i = y,W_C}$ impossible to sample from exactly, given that the players are not allowed to communicate and this distribution in general depends on both $x$ and $y$. A key insight in Raz' proof of parallel repetition is that it is still possible for the players to \emph{approximately} sample from the distribution using only local operations and shared randomness. 
Drawing on the work of Razborov~\cite{Razborov92} on analyzing the randomized communication complexity of the set disjointness problem, Raz~\cite{raz1998parallel} introduces a \emph{dependency-breaking random variable} $\bOmega$ with the following properties:
\begin{itemize}
\item[(a)] Given a sample $\bomega$ of the random variable $\bOmega$ both players can locally sample $\bx$ and $\by$ respectively such that the marginal distribution of $(\bx,\by)$ is approximately $\P_{\bX \bY | \bX_i = x, \bY_i = y,W_C}$,
\item[(b)] The players can jointly sample the same $\bomega \sim \bOmega$ using shared randomness.
\end{itemize}
In~\cite{Hol09} $\bOmega$ is defined so that a sample $\bomega$ fixes at least one of $\{x_{i'}, y_{i'}\}$ for each $i' \neq i$. It can then be shown that conditioned on $x$, $\bOmega$ is nearly (though not exactly) independent of $y$, and vice-versa. In other words, 
\begin{equation}
\label{eq:intro_cor_samp}
	\P_{\bOmega | \bX_i = x, W_C} \approx \P_{\bOmega | \bX_i = x, \bY_i = y, W_C} \approx \P_{\bOmega | \bY_i = y, W_C}
\end{equation}
where ``$\approx$'' denotes closeness in statistical distance. Eq.~\eqref{eq:intro_cor_samp}  suffices to guarantee that the players can \emph{approximately} sample the same $\bomega$ from $\P_{\bOmega | \bX_i = x, \bY_i = y, W_C}$ with high probability, achieving point (b) above. This sampling is accomplished through a technique called \emph{correlated sampling}. 

The proof of point (a) above in~\cite{raz1998parallel,Hol09} heavily relies on the assumption that the players employ a classical strategy (in fact, the analysis is done for the case of deterministic strategies). Quantum strategies have additional correlations due to the presence of entanglement between the players. The challenge is find an appropriate variant of (a) that applies in this case. Note that due to the possibility that the game $G$ has much better quantum strategies than the best classical, the rounding argument must necessarily result in a genuinely quantum strategy, so that there is no hope of a reduction to the classical case in general.

To formulate the issue more concretely, note that in order to execute the repeated strategy for $G^n$ conditioned on some questions $\bX_i = x$, $\bY_i = y$ received in $G$ (as well as the event $W_C$) it no longer suffices for the players to sample the vector of questions $(\bx,\by)$ from the distribution $\P_{\bX \bY | \bX_i = x, \bY_i = y,W_C}$. In the quantum case the probability space also involves the results of measurements performed by the players in $\strategy^n$ on a shared entangled state $\ket{\psi}$. Thus conditioning on the event $W_C$ also entails ``conditioning'' $\ket{\psi}$ on $W_C$ -- it is not clear \emph{a priori} what this means.

We model this as the requirement that the players in $G$ have access to \emph{some}  entangled state $\ket{\Phi_{x,y}}$ that ``simulates'' the appropriate conditional probability space. We call the state $\ket{\Phi_{x,y}}$ a \emph{dependency-breaking state}. It needs to satisfy two properties:
\begin{enumerate}
	\item \emph{Usefulness}: Given the shared state $\ket{\Phi_{x,y}}$ each player can perform a measurement depending on their question, $x$ or $y$, on their respective share of the state, such that the joint distribution of their measurement outcomes is close to $\P_{\bA_i \bB_i | \bX_i = x, \bY_i = y,W_C}$.
	\item \emph{Sampleability}: There exists a shared entangled state $\ket{\Phi}$, that is independent of $x$ and $y$, and unitary maps $U_x$ and $V_y$ such that $U_x \otimes V_y \ket{\Phi}$ is close to $\ket{\Phi_{x,y}}$, on average over $x,y$ sampled from the game distribution in $G$.
\end{enumerate}
The second property implies that, given question pair $(x,y)$, the players can start with the shared state $\ket{\Phi}$ and apply local operations to generate an approximation of $\ket{\Phi_{x,y}}$, in analogy with the classical correlated sampling procedure. Once the players have the approximation of $\ket{\Phi_{x,y}}$ the first property implies that they are able to generate   answers according to the right conditional distribution. 

It was shown by~\cite{jain2014parallel,chailloux2014parallel,chung2015parallel} that for the case of {free games} case (i.e.\ games with a product question distributions) it is possible to construct dependency-breaking states $\{ \ket{\Phi_{x,y}} \}_{x,y}$ and unitaries $U_x,V_y$ that satisfy both properties. However, establishing the existence of these states and unitaries for general  games is much more challenging, and indeed it remains an open problem to extend this approach to prove a quantum analogue of Raz's parallel repetition theorem (although we note that using this approach it is possible to prove a version with a polynomial decay bound~\cite{yuen2016parallel}).

\paragraph{Breaking correlations in repeated anchored games.} Our main contribution is to extend the framework of dependency-breaking variables and states to the setting of anchored games. We introduce dependency-breaking variables $\Omega$ and states $\ket{\Phi_{x,y}}$, and show that together they satisfy both Usefulness and Sampleability.

The analysis for anchored games is more intricate than for free games. Proofs of the analogous statements for free games in~\cite{jain2013parallel,  chailloux2014parallel, chung2015parallel} make crucial use of the fact that all possible question tuples are possible. An anchored game can be far from having this property. Instead, we use the anchors as a ``home base'' that is connected to all questions. Intuitively, no matter what question tuple $(x,y)$ we are considering, it is only a few replacements away from the set of anchor questions. Thus the dependency of the variable $\Omega$ on the questions can be iteratively removed by ``switching'' each player's question to an anchor as
$$
	\P_{\bOmega | \bX_i = x,\bY_i = y,W_C} \,\approx\, \P_{\bOmega | \bX_i = x,\bY_i = \dummy,W_C} \approx \P_{\bOmega | \bX_i = \dummy,\bY_i = \dummy,W_C}~.
$$

The dependency-breaking states $\ket{\Phi_{x,y}}$ we define are more complicated than those used in the free games case; in particular they also depend on the classical dependency-breaking variable $\bOmega$. To show that the states satisfy the Sampleability property we prove a sequence of approximations: first we show that for most $x$ there exists a unitary $U_x$ such that $(U_x \otimes \Id) \ket{\Phi_{\dummy, \dummy}} \approx \ket{\Phi_{x,\dummy}}$, where $\ket{\Phi_{\dummy, \dummy}}$ denotes the dependency-breaking state in the case that both players receive the anchor question ``$\dummy$'', and $\ket{\Phi_{x,\dummy}}$ denotes the state when the first player receives $x$ and the second player receives ``$\dummy$''. Then we show that on average over $y$ there exists a unitary $V_y$ such that $(\Id \otimes V_y) \ket{\Phi_{x,\dummy}} \approx \ket{\Phi_{x,y}}$. Interestingly a crucial component of our proof is to argue the existence of a local unitary $R_{x,y}$ that depends on \emph{both} inputs $x$ and $y$. The unitary $R_{x,y}$ is not implemented by either player in the strategy for the single-shot game $G$, but it is needed to show that $V_y$ maps $\ket{\Phi_{x,\dummy}}$ close to $\ket{\Phi_{x,y}}$. 
 \section{Preliminaries}
\label{sec:prelim}

\subsection{Sets and indices} 
For an integer $n$ we let $[n] = \{1,\ldots,n\}$. For a finite set $\X$ we let $\X^n$ denote the $n$-fold Cartesian product of $\X$. We denote elements of $\X^n$ by boldfaced letters $\bx = (\bx_1,\ldots,\bx_n)$. For a subset $C = \{i_1,\ldots,i_t\} \subseteq [n]$ we let $\bx_C$ denote the ordered tuple $(\bx_{i_1},\ldots,\bx_{i_t})$.

\subsection{Probability distributions, random variables, and expectations}\label{subsec:prob_dist}

We use capital letters to denote random variables and lower case letters to denote values taken by the random variable. We use boldfaced letters to denote tuples, e.g.\ $\bX = (\bX_1,\ldots,\bX_n)$, $\bx = (\bx_1,\ldots,\bx_n)$. For a subset $C \subseteq [n]$ we write $\bX_C$ to denote the $|C|$-tuple formed by  coordinates of $\bX$ indexed by $C$. 

We use $\P_X$ to denote the distribution of random variable $X$ and $\P_X(x)$ to denote the probability that $X = x$ for some value $x$. For multiple random variables, e.g.\ $X, Y, Z$, $\P_{XYZ}(x,y,z)$  denotes their joint distribution. All random variables are assumed to operate on the same probability space, which is usually implicit and clear from context. 

We use $\P_{Y | X = x}(y)$ to denote the conditional distribution $\P_{YX}(y,x)/\P_X(x)$, which is defined when $\P_X(x) > 0$. We use the shorthand $\P_{X | y,z}$ to denote the distribution $\P_{X | Y =y,Z=z}$. For example, we may write $\P_{V | \omega_\mi, \bx_i, \by_i}$ to denote $\P_{V | \Omega_\mi = \omega_\mi, \bX_i = \bx_i, \bY_i = \by_i}$. For an event $W$ we let $\P_{X Y | W}$ denote the distribution conditioned on $W$. We use the notation $\Ex_{X} f(x)$ and $\Ex_{\P_X} f(x)$ to denote the expectation $\sum_{x} \P_X(x) f(x)$. 

Let $\P_{XY}$ be a joint distribution on $\X \times \Y$ and let $W$ denote an event. Then we define the distribution $\P_{X|W} \P_{Y|X}$ over $\X \times \Y$ as
\[
	(\P_{X|W} \P_{Y|X})(x,y) = \P_{X|W}(x) \cdot \P_{Y | X = x}(y)~.
\]

For distributions $P_X$ and $P_Y$ over the same set $\X$ we use $\| \P_{X_0} - \P_{X_1} \|$ to denote their total variation distance, 
$$\| \P_{X_0} - \P_{X_1} \| \,=\, \frac{1}{2}\sum_{x \in \X} |\P_{X_0}(x) - \P_{X_1} (x)|\;.$$

The following simple lemmas will be frequently used. 

\begin{lemma}\label{lem:trivial}
Let $\Qsf_F$ and $\Ssf_F$ be two probability distributions for  random variable $F$ and let $\Rsf_{G | F}$ be a conditional probability distribution for random variable $G$, conditioned on $F$. Then
\[ \big\| \Qsf_{F} \Rsf_{G|F} - \Ssf_{F} \Rsf_{G|F}\big\|\,=\, \big\|\Qsf_{F}-\Ssf_{F}\big\|\;. \]
Similarly, for two conditional probability distributions $\Qsf_{G|F}, \Ssf_{G|F}$ and a distribution $\Rsf_F$, 
\[
\big\| \Rsf_F \Qsf_{G|F} - \Rsf_F \Ssf_{G|F} \big\| \,=\, \Ex_F \big\| \Qsf_{G|F = f} - \Ssf_{G|F=f} \big\|\;,
\]
where $\Ex_F$ denotes the expectation over sampling $f$ from $\Rsf_F$.
\end{lemma}

\begin{proof}
For the first part of the lemma using the definition we note that $ \| \Qsf_{F} \Rsf_{G|F} - \Ssf_{F} \Rsf_{G|F}\|$ is equal to
\begin{align*}
 \frac{1}{2}\sum_{f,g} | \Qsf(f) \Rsf(g|f)- \Ssf(f) \Rsf (g|f)|&=  \frac{1}{2}\sum_{f} |\Qsf(f)- \Ssf(f)| \cdot  \Big( \sum_{g} \Rsf(g|f)  \Big)\\
&= \frac{1}{2}\sum_{f} | \Qsf(f)- \Ssf(f)|\\
&= \|\Qsf_{F}-\Ssf_{F}\|\;.
\end{align*}
For the second part of the lemma we note that $\| \Rsf_F \Qsf_{G|F} - \Rsf_F \Ssf_{G|F} \|$ is equal to
\begin{align*}
\frac{1}{2}\sum_{f,g}	 | \Rsf(f) \Qsf(g|f) - \Rsf(f) \Ssf(g|f) | &= \sum_f \Rsf(f) \frac{1}{2} \sum_g	 | \Qsf(g|f) - \Ssf(g|f) | \\
&= \sum_f \Rsf(f) \| \Qsf_{G|F = f} - \Ssf_{G|F=f} \|\;.
\end{align*}
\end{proof}

\begin{lemma}[Data processing inequality]\label{lem:data-processing}
Let $\Qsf_{FG}$ and $\Ssf_{FG}$ denote two probability distributions for random variables $F,G$. Then 
\[
	\big \| \Qsf_{F} - \Ssf_F \big \| \leq \big \| \Qsf_{FG} - \Ssf_{FG} \big \|~.
\]
\end{lemma}
\begin{proof}
Expanding the definition of $\big \| \Qsf_{F} - \Ssf_F \big \|$, we get
	\begin{align*}
		\frac{1}{2} \sum_f  \big | \Qsf_F(f) - \Ssf(f) \big | &= \frac{1}{2} \sum_f \big | \sum_g \big( \Qsf_F(f) \cdot \Qsf_{G|F = f}(g) - \Ssf_F(f) \cdot \Ssf_{G|F=f}(g) \big) \big| \\
										&\leq \frac{1}{2} \sum_{f,g} \big | \Qsf_F(f) \cdot \Qsf_{G|F = f}(g) - \Ssf_F(f) \cdot \Ssf_{G|F=f}(g) \big| \\
										&= \big \| \Qsf_{FG} - \Ssf_{FG} \big \|
	\end{align*}
	where the second line follows from the triangle inequality.
\end{proof}

\subsection{Quantum states and measurements}

All Hilbert spaces considered in this paper are finite-dimensional. We write $\Id$ to denote the identity operator. For Hermitian matrices $A, B$ we write $A \leq B$ to indicate that $A - B$ is positive semidefinite. For a linear operator $X$ acting on a complex vector space we use $X[\cdot]$ to denote the map that takes linear operators $\rho \mapsto X\rho X^\dagger$.
For a vector $\ket{\psi}$,we use $\| \ket{\psi} \|$ to denote its Euclidean length. For a matrix $A$ we use $\| A \|_1$ to denote its \emph{trace norm} $\Tr(\sqrt{A A^\dagger})$. 

A density matrix is a positive semidefinite matrix with trace $1$. A positive operator-valued measure (POVM) $M$ acting on $\C^d$ with outcome set $\A$ is denoted as a function $M$ mapping outcomes $a \in \A$ to positive semidefinite operators acting on $\C^d$ such that $\sum_{a\in\A} M(a) = \Id$. 

We use the convention that, when $\ket{\psi}$ is a pure state, $\psi$ refers to the rank-1 density matrix $\ketbra{\psi}{\psi}$. We use superscripts to label subsystems. For example, $\rho^{AB}$ denotes a density matrix on systems $A$ and $B$, i.e.\ $\rho^{AB} \in \mH_A\otimes \mH_B$ for finite-dimensional Hilbert spaces $\mH_A$ and $\mH_B$ that will always be clear from context. $\rho^A$ and $\rho^B$ are used to denote the partial trace of $\rho^{AB}$ on the first and second subsystem respectively. A \emph{classical-quantum (CQ)} state $\rho^{XE}$ is classical on $X$ and quantum on $E$ if it can be written as $\rho^{XE} = \sum_{x} p(x) \ketbra{x}{x}^X \otimes \rho^E_x$ for $\{\ket{x}\}$ the canonical basis of system $X$ and some probability measure $p(\cdot)$. We write $\rho^{XE}_x$ to denote the state $\ketbra{x}{x}^X \otimes \rho^E_x$.

\subsection{Nonlocal games}

\begin{definition}[Nonlocal game]
A \emph{nonlocal game} $G$ (or \emph{game} for short) is a tuple $(\X \times \Y,\A \times \B,\mu,V)$ where $\X, \Y$ are finite sets called the \emph{question sets}, $\A, \B$ are finite sets called the \emph{answer sets}, $\mu$ is a distribution over $\X \times \Y$ called the \emph{question distribution}, and $V: \X \times \Y \times \A \times \B \to \{0,1\}$ is a predicate called the \emph{game predicate}.
\end{definition}

\begin{definition}[Parallel repetition of a nonlocal game]
Let $G = (\X \times \Y,\A \times \B,\mu,V)$ be a nonlocal game and $n\geq 1$ an integer. The \emph{$n$-fold parallel repetition} of $G$ is the nonlocal game $G^n = (\X^n \times \Y^n, \A^n \times \B^n,\mu^n,V^n)$ where $\mu^n$ is the product distribution $\mu \times \cdots \times \mu$ over $(\X \times \Y)^n$ and $V^n: \X^n \times \Y^n \times \A^n \times \B^n \to \{0,1\}$ is the predicate defined as
\[
	V^n(\bx,\by,\ba,\bb) \,=\, \prod_{i=1}^n V(\bx_i,\by_i,\ba_i,\bb_i)
\]
for all $\bx \in \X^n,\by \in \Y^n,\ba \in \A^n,\bb \in \B^n$.
\end{definition}

\begin{definition}[Quantum strategy]
Let $G = (\X \times \Y,\A \times \B,\mu,V)$ be a nonlocal game. A \emph{quantum strategy $\strategy$ for the game $G$} is a tuple $(\ket{\psi}, A,B)$ where $\ket{\psi}$ is a state in $\C^d \otimes \C^d$ for some positive integer $d$ and $A = \{ A_x \}_{x \in \X}$ and $B = \{B_y\}_{y \in \Y}$ are collections of POVMs (with outcomes in $\A$ and $\B$, respectively) acting on $\C^d$. The \emph{dimension of $\strategy$} is $d$. 
\end{definition}

\begin{definition}[Quantum value of a nonlocal game]
Let $G = (\X \times \Y,\A \times \B,\mu,V)$ be a nonlocal game and $\strategy$ a strategy for $G$. The \emph{quantum value of the strategy $\strategy$ in the game $G$} is defined as
\[
	\eval(G,\strategy) \,=\, \sum_{(x,y) \in \X \times \Y} \mu(x,y) \,  \sum_{(a,b) \in \A \times \B} V(x,y,a,b) \, \bra{\psi} A_x(a) \otimes B_y(b) \ket{\psi}\;.
\]
The \emph{quantum value of $G$} is defined as
\[
	\eval(G) \,=\, \sup_{\strategy} \eval(G,\strategy)\;,
\]
where the supremum is over the set of all quantum strategies for $G$.
\end{definition}

\begin{definition}[Entanglement requirements of a nonlocal game]
Let $G = (\X \times \Y,\A \times \B,\mu,V)$ be a nonlocal game, and let $p \in [0,1]$. Define $\E(G,p)$ to denote the minimum integer $d \in \N$ such that there exists a $d$-dimensional quantum strategy $\strategy$ for $G$ such that $\eval(G,\strategy) \geq p$. If there is no such strategy then we define $\E(G,p) = +\infty$.
\end{definition}

\section{Dependency-breaking variables, states, and measurements}
\label{sec:quantum_setup}

We introduce random variables, entangled states and operators that will be used in the proof of Theorem~\ref{thm:anchorpr_quantum}. For the entirety of the section we fix an $0<\alpha < 1$ and an $\alpha$-anchored two-player game $G = (\X \times \Y,\A \times \B,\mu,V)$. We further fix an integer $n\geq 1$ and a strategy $\strategy^n = (\ket{\psi},A,B)$ for the parallel-repeated game $G^n$. By conjugating the second player's measurement operators by a unitary $U$ and exchanging $\ket{\psi}$ with $\Id \otimes U \ket{\psi}$ we may assume without loss of generality that $\ket{\psi}$ can be written as 
\begin{equation}\label{eq:psi-sym}
 \ket{\psi} \,=\, \sum_j \sqrt{\lambda_j} \ket{v_j} \ket{v_j}\,\in \, \C_{E_A}^d \otimes \C_{E_B}^d\;,
\end{equation}
for some orthonormal basis $\{ \ket{v_j} \}_j$ of $\C^d$. Here we used $d$ to denote the dimension of the strategy $\strategy$ and label the two systems on which $\ket{\psi}$ lies $E_A$ and $E_B$ respectively. 

This section is divided into three parts. First we define some random variables and their joint distribution $\P$. Then we state useful lemmas about conditioned distributions. Finally we describe states and operators used in the proof of Theorem~\ref{thm:anchorpr_quantum}.

\subsection{The probability measure $\P$} 
\label{sec:p_setup}

Let $\mu_X$ and $\mu_Y$ denote the marginals of the game distribution $\mu_{XY}$ on the first and second coordinates, respectively. 
We first define a ``single copy'' distribution $\hat{P}$ as the law of random variables $(D,M,X,Y)$ which are defined as follows. Each random variable may depend on previously defined ones. We start with $D$, which is distributed uniformly over $\{A,B\}$. Fix a ``noise'' parameter 
\begin{equation}
\label{eq:def-eta}
\eta \,=\, \frac{\alpha}{2}\;.
\end{equation}
Let $M$ have the following distribution over $\X\times \Y$: for all $x \in \X$ and $y \in \Y$, 
\[
	\P_{M | D = A}(x) = \left\{
	\begin{array}{ll}
		\frac{\mu_X(x)}{1 - \eta}  & \mbox{if $x \neq \dummy$} \\[4mm]
		\frac{\alpha - \eta}{1 - \eta} & \mbox{if $x = \dummy$ }
	\end{array}
\right. \quad \text{and} \quad
	\P_{M | D = B}(y) = \left\{
	\begin{array}{ll}
		\frac{\mu_Y(y)}{1 - \eta}  & \mbox{if $y \neq \dummy$} \\[4mm]
		\frac{\alpha - \eta}{1 - \eta} & \mbox{if $y = \dummy$ }
	\end{array}
\right. \;.
\]
In other words, conditioned on $D=A$ (resp. $D = B$), the variable $M$ takes on a value in $\X$ (resp. $\Y$) from a rescaled version of the distribution $\mu_X$ (resp. $\mu_Y$) where less weight is given to the dummy question $\dummy$. Finally, define the random variables $(X,Y)$ as follows. 
\begin{itemize}
	\item If $D = A$ then $X$ is chosen to be an ``$\eta$-noisy'' copy of $M$. Precisely, $X = M$ with probability $1 - \eta$ and $X = \dummy$ with probability $\eta$. Define $Y$ to equal $y$ with probability $\mu_{Y|X}(y | m)$, where $m$ is the value of $M$. 
	\item If $D = B$ then $Y$ is chosen to be an ``$\eta$-noisy'' copy of $M$ and $X$ equals $x$ with probability $\mu_{X|Y}(x | m)$, where $m$ is the value of $M$. 
\end{itemize}
This specifies the distribution $\hat{P}$. 
The following claims will be frequently used. 

\begin{claim}
\label{clm:dependency-breaking-1}
	Conditioned on $(D,M)$ the random variables $X$ and $Y$ are independent. 
\end{claim}
\begin{proof}
It follows directly from the construction that for every $(d,m)$ and $(x,y)$, 
\[
		\hat{\P}_{XY | D=d,M=m}(x,y) \,=\, \hat{\P}_{X | D=d,M=m}(x) \cdot \hat{\P}_{Y | D=d,M=m}(y)\;.\]
\end{proof}

\begin{claim}
\label{clm:dependency-breaking-2}
 $\hat{\P}_{XY|D = A}(x,y) = \hat{\P}_{XY|D=B}(x,y) = \mu_{XY}(x,y)$ for all $x \in \X,y \in \Y$. In particular, the marginal distribution $\hat{\P}_{XY}$ is identical to the game distribution $\mu$. 
\end{claim}

\begin{proof}
	Fix $D = A$ and suppose that $x = \dummy$. Using the conditional independence stated in Claim~\ref{clm:dependency-breaking-1} we can write 
	\begin{align*}
		\hat{\P}_{XY | D = A}(\dummy,y) &= \sum_{x' \in \X} \hat{\P}_{M | D = A}(x') \cdot \hat{\P}_{X | M=x',D=A}(\dummy) \cdot  \hat{\P}_{Y | M = x', D = A}(y)  \\
		&= \hat{\P}_{M | D = A}(\dummy) \cdot \hat{\P}_{X | M = \dummy, D = A}(\dummy) \cdot  \hat{\P}_{Y | M = \dummy, D = A}(y) \\
		& \qquad \qquad + \sum_{x' \in \X \setminus \{\dummy\}} \hat{\P}_{M | D = A}(x') \cdot \hat{\P}_{X | M = x', D = A}(\dummy) \cdot  \hat{\P}_{Y | M = x, D = A}(y) \\
		&= \frac{\alpha - \eta}{1 - \eta} \cdot \mu_Y(y) + \frac{\eta}{1 - \eta} \sum_{x' \in \X \setminus \{\dummy\}} \mu_X(x') \cdot \mu_{Y|X}(y|x') \\
		&= \frac{\alpha - \eta}{1 - \eta} \cdot \mu_Y(y) + \frac{\eta}{1 - \eta} \Big ( \mu_Y(y) - \mu_{XY}(\dummy,y) \Big ) \\
		&= \frac{\alpha - \eta}{1 - \eta} \cdot \mu_Y(y) + \frac{\eta}{1 - \eta} \Big ( \mu_Y(y) - \mu_{Y}(y) \cdot \alpha \Big ) \\
		&= \alpha \cdot \mu_Y(y) \\
		&= \mu_{XY}(\dummy,y)~.
	\end{align*}
	Now suppose that $x \neq \dummy$. The random variable $X$ can only take this value if $M = x$, so we have
	\begin{align*}
		\hat{\P}_{XY | D = A}(x,y) &= \hat{\P}_{M | D = A}(x) \cdot \hat{\P}_{X | M = x, D = A}(x) \cdot  \hat{\P}_{Y | M=x, D=A}(y) \\
		&= \frac{1}{1 - \eta} \cdot \mu_X(x) \cdot (1 - \eta) \cdot \mu_{Y|X}(y | x) \\
		&= \mu_{XY}(x,y)~.
	\end{align*}
	A symmetric argument holds for the case $D = B$. This shows the claim.
\end{proof}

We define a distribution $\P$ as follows. Let $\bD = (\bD_1,\ldots,\bD_n)$, $\bM = (\bM_1,\ldots,\bM_n)$, $\bX = (\bX_1,\ldots,\bX_n)$, and $\bY = (\bY_1,\ldots,\bY_n)$ be vectors of random variables, and define
\[
	\P_{\bD \bM \bX \bY} \,=\, \prod_{i=1}^n \hat{\P}_{\bD_i \bM_i \bX_i \bY_i}~.
\]
Finally, define random variables $\bA = (\bA_i)_{i \in [n]}$ and $\bB = (\bB_i)_{i \in [n]}$ as follows. Conditioned on $\bX$ and $\bY$, the random variables $\bA,\bB$ are independent of $\bD$ and $\bM$. For all realizations $\bx, \by$ of $\bX$ and $\bY$, define
\[
	\P_{\bA \bB | \bX = \bx, \bY = \by}(\ba,\bb) \,=\, \bra{\psi} A_\bx(\ba) \otimes B_\by(\bb) \ket{\psi}
\]
for all $\ba = (\ba_1,\ldots,\ba_n) \in \A^n$ and $\bb = (\bb_1,\ldots,\bb_n) \in \B^n$. $\bA$ and $\bB$ are jointly distributed as answers produced by players using strategy $\strategy^n$ when their question pair for the $i$-th game is $(\bx_i,\by_i)$. 

\begin{claim}\label{claim:abxy}
The marginal distribution of $\bX \bY \bA \bB$ is identical to the the joint distribution of questions and answers obtained in the repeated game $G^n$ when the players use the strategy $\strategy^n$.
\end{claim}

\begin{proof}
Since by definition conditioned on $\bX, \bY$, $\bA$ and $\bB$ are independent of $\bD,\bM$ we have that the marginal distribution of $\bX \bY \bA \bB$ is
	\[
		\P(\bx,\by,\ba,\bb) = \Big ( \prod_{i=1}^n \mu(\bx_i,\by_i) \Big) \cdot \bra{\psi} A_\bx(\ba) \otimes B_\by(\bb) \ket{\psi}\;,
	\]
	which matches the joint distribution of questions and answers in the repeated game $G^n$ when the players use the strategy $\strategy^n$.
	\end{proof}

\subsection{Dependency-breaking variables} 
\label{sec:dep-var}

We introduce \emph{dependency-breaking variables}. These are crucial for controlling the correlations that arise when conditioning the  distribution $\P$ on different events. 

Let $C \subseteq [n]$ be a set of coordinates for the repeated game $G^n$. 
Let $(\bX_C,\bY_C)$ and $(\bA_C,\bB_C)$ be random variables associated with the players' questions and answers in the coordinates indexed by $C$. For $i \in [n]$ let $W_i$ denote the indicator variable for the event that the players win round $i$. Using~Claim~\ref{claim:abxy}, $W_i$ is the event that $V(\bX_i,\bY_i,\bA_i,\bB_i) = 1$ where $\bX,\bY,\bA,\bB$ are the random variables defined in Section~\ref{sec:p_setup}. Let $W_C = \prod_{i \in C} W_i$. 

\begin{definition}
For all $i \in [n] \setminus C$ define the \emph{$i$-th dependency-breaking variable $\bOmega_i$} as 
\[\bOmega_i \,=\, (\bD_i,\bM_i)\;.\]
Further define 
\[ \bOmega \,=\, (\bOmega_i)_{i \in [n] \setminus C} \cup (\bX_C,\bY_C)\qquad \text{and}\qquad \bOmega_\mi\,=\,(\bOmega_j)_{j \in [n] \setminus (C \cup \{i\})} \cup (\bX_C,\bY_C) \;,\quad \forall i \in [n] \setminus C\;.\]
Finally, define variables
\[\bR \,=\, (\bOmega,\bA_C,\bB_C)\qquad\text{and}\qquad\bR_\mi \,=\, (\bOmega_\mi,\bA_C,\bB_C)\;,\qquad \forall i\in[n]\backslash C\;.\]
\end{definition}

 When $\eta = 0$ the definition of $\bOmega_i$ coincides with the one used by Holenstein~\cite{Hol09}; in that case, the variable $\bM_i$ is coupled to either $\bX_i$ or $\bY_i$ exactly. Here we set $\eta$ to be a nonzero value that is related to $\alpha$, the anchoring probability, as in~\eqref{eq:def-eta}. This ``noisy coupling'' between $\bOmega_i$ and the inputs $(\bX_i,\bY_i)$ is important for our analysis. 

We denote realizations of the random variable $\bOmega$ using $\bomega = (\bomega_i)_{i \in [n] \setminus C} \cup (\bx_C,\by_C)$, and $\bOmega_\mi$ using $\bomega_\mi$. Similarly we denote realizations of $\bR$ as $\br = (\bomega,\ba_C,\bb_C)$ and $\bR_\mi$ as $\br_\mi = (\bomega_\mi,\ba_C,\bb_C)$.

\begin{claim}\label{claim:ufacts}
The following properties hold for the distribution $\P$.
\begin{enumerate}
\item The joint distribution $(\bX,\bY,\bOmega)$ is product across its $n$ triples of coordinates. Furthermore, for any $i$, $\bX_i$ and $\bY_i$ are independent conditioned on $\bOmega_i$. In particular,
 $\P_{\bOmega_i \bX_i \bY_i} =\P_{\bOmega_i} \P_{\bX_i | \bOmega_i} \P_{\bY_i | \bOmega_i}$. 
\item $\P_{\bA_C \bB_C | \bX \bY} = \P_{\bA_C \bB_C | \bOmega \bX \bY}$.
\end{enumerate}
\end{claim}

\begin{proof} 
\hfill
\begin{enumerate}
\item The first part is by construction, and the second part by Claim~\ref{clm:dependency-breaking-1} and the definition of $\bOmega_i$. 
\item By definition, conditioned on $\bX$ and $\bY$ the random variables $(\bA ,\bB)$ are independent from $(\bD,\bM)$ and thus from $\bOmega = (\bD,\bM)_{i\in [n]\setminus C} \cup (\bX_C,\bY_C)$. 
\end{enumerate}
\end{proof}
\subsection{Individual coordinates are relatively unaffected by conditioning} 

Let $C\subseteq [n]$ be a fixed set of coordinates such that $\P(W_C) > 0$, and let $m=n-C$. For convenience, up to relabeling we assume that $C$ contains the last $n-m$ coordinates $C=\{m+1,\ldots,n\}$. Define
\begin{equation}
\label{eq:delta}
	\delta = \frac{1}{m} \left ( \log \frac{1}{\P(W_C)} + |C| \log |\A| |\B| \right)\;.
\end{equation}

The following lemma shows that if the event $W_C$ occurs with significant probability then conditioning on $W_C$ only has a moderate effect on the distribution of $(\bX_i, \bY_i)$, on average over a uniformly random choice of $i \in [n] \setminus C$. Furthermore, the distribution of $\bR_\mi$ is close to being independent from $(\bX_i,\bY_i)$. 

\begin{lemma}
\label{lem:classical_skew}
The following inequalities hold:
\begin{enumerate}
	\item $\frac{1}{m}\sum_{i=1}^m \| \P_{\bOmega_i \bX_i \bY_i | W_C} -  \P_{\bOmega_i \bX_i \bY_i} \| \leq \sqrt{\delta}$.
	\item $\frac{1}{m}\sum_{i=1}^m\|  \P_{\bR \bX_i \bY_i | W_C} -  \P_{\bR |W_C} \P_{\bX_i \bY_i | \bOmega_i } \| \leq \sqrt{\delta}$.
	\item $\frac{1}{m}\sum_{i=1}^m\| \P_{\bOmega_i |W_C} \P_{\bR_\mi | \bX_i = \dummy, \bY_i = \dummy, W_C} -\P_{\bOmega_i | W_C} \P_{\bR_\mi |\bOmega_i  W_C}  \big \| = O(\sqrt{\delta}/\alpha^2)$.
	\item $\frac{1}{m}\sum_{i=1}^m\big\|  \P_{\bX_i \bY_i}  \P_{\bR_\mi | \bX_i = \dummy, \bY_i = \dummy, W_C} -  \P_{\bX_i \bY_i}  \P_{\bR_\mi | \bX_i, \bY_i, W_C} \big\| = O(\sqrt{\delta}/\alpha^2)$.
\end{enumerate}
\end{lemma}

The proof of Lemma~\ref{lem:classical_skew} makes use of Lemma 4.1 and Corollary 4.3 from~\cite{Hol09}, which we restate for convenience. 

\begin{lemma}[Lemma 4.1 of~\cite{Hol09}]
\label{lem:hol}
Let $\bU = (\bU_1,\ldots,\bU_m)$ be a vector of random variables and let $\P_{\bU} = \P_{\bU_1} \cdots \P_{\bU_m}$. Let $W$ denote an event. Then
\[
	\sum_{i=1}^m \, \big \| \P_{\bU_i} - \P_{\bU_i | W} \big \|^2 \,\leq \, \log \frac{1}{\P(W)}~.
\]
\end{lemma}

\begin{corollary}[Corollary 4.3 of~\cite{Hol09}]
\label{cor:hol}
Let $T$, $V$ be random variables and let $\bU = (\bU_1,\ldots,\bU_m)$ be a vector of random variables. Let $\P_{T\bU V} = \P_T \cdot \P_{\bU_1 | T} \cdots \P_{\bU_m} \cdot \P_{V|T \bU}$. Let $W$ denote an event. Then
\[
	\frac{1}{m} \sum_{i=1}^m \big \| \P_{T \bU_i V|W} - \P_{TV|W} \P_{\bU_i | T} \big \| \,\leq\, \sqrt{ \frac{1}{m} \Big ( \log(|\mathcal{V}^*|) + \log \frac{1}{\P(W)}  \Big)}\;,
\]
where $\mathcal{V}^* = \{ v : \P_{V|W}(v) > 0 \}$. 
\end{corollary}

\begin{proof}[Proof of \Cref{lem:classical_skew}]
\hfill
\begin{enumerate}
\item For $i\in [n]\backslash C$ let $\bU_i$ denote the tuple $(\bOmega_i,\bX_i,\bY_i)$, and let $W$ denote the event $W_C$. Note that all of the $\bU_i$ are independent of each other by construction, and thus applying \Cref{lem:hol} we get
\[
\sum_{i =1}^m \, \big \| \P_{\bOmega_i \bX_i \bY_i} - \P_{\bOmega_i \bX_i \bY_i | W_C} \big \|^2 \leq \log \frac{1}{\P(W_C)}~.
\]
Dividing by $m$ on both sides and using Jensen's inequality we get
\begin{align*}
\frac{1}{m} \sum_{i=1}^m \, \big \| \P_{\bOmega_i \bX_i \bY_i} - \P_{\bOmega_i \bX_i \bY_i | W_C} \big \| &\leq \sqrt{ \frac{1}{m} \sum_{i=1}^m \, \big \| \P_{\bOmega_i \bX_i \bY_i} - \P_{\bOmega_i \bX_i \bY_i | W_C} \big \|^2} \\
&\leq \sqrt{ \frac{1}{m} \log \frac{1}{\P(W_C)}}\,
\end{align*}
which by definition of $\delta$ in~\eqref{eq:delta} is at most $\sqrt{\delta}$. \item Let $T$ denote the random variable $\bOmega$, let $\bV_i$ denote the pair $(\bX_i,\bY_i)$, and let $V$ denote the pair $(\bA_C,\bB_C)$. Since $\bV_1,\ldots,\bV_m$ are independent, even conditioned on $\bOmega$ which has a product distribution, we can apply 
\Cref{cor:hol} and get
\begin{align}
	\frac{1}{m} \sum_{i=1}^m \big \| \P_{\bR \bX_i \bY_i|W_C} - \P_{\bR|W_C} \P_{\bX_i \bY_i | \bOmega} \big \| &\leq \sqrt{ \frac{1}{m} \left ( \log \Big ( (|\A| \cdot |\B|)^{|C|} \Big) + \log \frac{1}{\P(W_C)}  \right)} \notag\\
	&= \sqrt{\delta}\;,\label{eq:classical-skew-1}
\end{align}
where we used the fact that the support of $V$ has size at most $(|\A| \cdot |\B|)^{|C|}$. We then use that $\P_{\bX_i \bY_i | \bOmega} = \P_{\bX_i \bY_i | \bOmega_i}$ to obtain the second item.
\item We start by rewriting~\eqref{eq:classical-skew-1} using Bayes' rule to obtain
\begin{equation}
\label{eq:classical-skew-2}
\frac{1}{m} \sum_{i=1}^m \big \| \P_{\bOmega_i |W_C} \P_{\bX_i \bY_i | \bOmega_i W_C} \P_{\bR_\mi | \bOmega_i \bX_i \bY_i W_C} - \P_{\bOmega_i | W_C} \P_{\bX_i \bY_i | \bOmega_i} \P_{\bR_\mi |\bOmega_i  W_C}  \big \| \leq \sqrt{\delta}\;,
\end{equation}
where we used that by definition $\bR = (\bOmega_i,\bR_\mi)$ and for the second term that the joint distribution of $(\bX,\bY,\bOmega)$ is product across all $n$ coordinates.  
Next we show that 
\begin{equation}
\label{eq:classical-skew-2c}
\P_{\bR_\mi | \bOmega_i \bX_i \bY_i W_C} = \P_{\bR_\mi |\bX_i \bY_i W_C}\;.
\end{equation}
This follows from repeatedly applying Bayes' rule. In detail, for fixed $\br_\mi$ that implies the event $W_C$ \footnote{Since $\bR_\mi$ includes the random variables $(\bX_C,\bY_C,\bA_C,\bB_C)$, the event $W_C$ is determined by $\bR_\mi$.} and fixed $\bomega_i, \bx_i, \by_i$, 
\begin{align}
	\P_{\bR_\mi | \bomega_i, \bx_i, \by_i, W_C}(\br_\mi) &= \frac{\P_{\bR_\mi | \bomega_i, \bx_i, \by_i}(\br_\mi) \cdot \P(W_C | \bomega_i, \bx_i, \by_i, \br_\mi)}{\P(W_C | \bomega_i, \bx_i, \by_i) } \notag \\
	&= \frac{\P_{\bR_\mi | \bomega_i, \bx_i, \by_i}(\br_\mi)}{\P(W_C | \bomega_i, \bx_i, \by_i) }\;, \label{eq:classical-skew-2b}
\end{align}
where we used that since $\br_\mi$ implies the event $W_C$, $\P(W_C | \bomega_i, \bx_i, \by_i, \br_\mi) = 1$. Letting $\br_\mi = (\bomega_\mi,\ba_C,\bb_C)$, we have
\begin{align*}
	\P_{\bR_\mi | \bomega_i, \bx_i, \by_i}(\br_\mi) &= \P_{\bOmega_\mi | \bomega_i,\bx_i, \by_i}(\bomega_\mi) \cdot \sum_{\bx_\mi,\by_\mi} \P_{\bX_\mi \bY_\mi | \bomega, \bx_i, \by_i}(\bx_\mi,\by_\mi) \cdot \P_{\bA_C \bB_C | \bomega, \bx, \by}(\ba_C,\bb_C) \\
	&= \P_{\bOmega_\mi | \bx_i, \by_i}(\bomega_\mi) \cdot \sum_{\bx_\mi,\by_\mi} \P_{\bX_\mi \bY_\mi | \bomega_\mi, \bx_i, \by_i}(\bx_\mi,\by_\mi) \cdot \P_{\bA_C \bB_C | \bx, \by}(\ba_C,\bb_C) \\
	&= \P_{\bR_\mi | \bx_i, \by_i}(\br_\mi)\;,
\end{align*}
where in the second line we used that $\bOmega_\mi$, $\bX_\mi$, $\bY_\mi$ are independent of $\bOmega_i$ and Item 2 of \Cref{claim:ufacts}. This implies that 
\[
\P(W_C | \bomega_i, \bx_i,\by_i) = \sum_{\br_\mi \in W_C} \P_{\bR_\mi | \bomega_i,\bx_i,\by_i} (\br_\mi) = \sum_{\br_\mi \in W_C} \P_{\bR_\mi | \bx_i,\by_i} (\br_\mi) = \P(W_C | \bx_i,\by_i)\;,
\]
where $\br_\mi \in W_C$ denotes all $\br_\mi$ that imply the event $W_C$. Combined with~\eqref{eq:classical-skew-2b} this shows~\eqref{eq:classical-skew-2c}.

Item 1 of the present lemma combined with the data processing inequality (\Cref{lem:data-processing}) implies that 
\begin{equation}\label{eq:classical-skew-2bb}
\frac{1}{m} \sum_{i=1}^m \big\| \P_{\bOmega_i | W_C} - \P_{\bOmega_i}  \big\| \leq \sqrt{\delta}\;.
\end{equation}
Using \Cref{lem:trivial} we get that $\frac{1}{m} \sum_{i=1}^m \| \P_{\bOmega_i } \P_{\bX_i \bY_i | \bOmega_i} - \P_{\bOmega_i | W_C} \P_{\bX_i \bY_i | \bOmega_i} \| \leq \sqrt{\delta}$ and thus combined with Item 1 and the triangle inequality we have 
\[
\frac{1}{m} \sum_{i=1}^m \| \P_{\bOmega_i | W_C} \P_{\bX_i \bY_i | \bOmega_i W_C} - \P_{\bOmega_i | W_C} \P_{\bX_i \bY_i | \bOmega_i} \| \leq 2\sqrt{\delta}~.
\]
\Cref{lem:trivial} then implies that~\eqref{eq:classical-skew-2} is within at most $2\sqrt{\delta}$ of
\begin{equation}
\label{eq:classical-skew-2d}
\frac{1}{m} \sum_{i=1}^m \big \| \P_{\bOmega_i |W_C} \P_{\bX_i \bY_i | \bOmega_i} \P_{\bR_\mi | \bX_i \bY_i W_C} - \P_{\bOmega_i | W_C} \P_{\bX_i \bY_i | \bOmega_i} \P_{\bR_\mi |\bOmega_i  W_C}  \big \|\;.
\end{equation}
Notice that conditioned on $\bOmega_i$, by construction the variables $(\bX_i,\bY_i)$ take on the value $(\dummy,\dummy)$ with probability at least $\eta^2 = \Omega(\alpha^2)$. Thus conditioning both sides of the difference in~\eqref{eq:classical-skew-2d} on $(\bX_i,\bY_i) = (\dummy,\dummy)$ we get that 
\begin{equation}
\label{eq:classical-skew-3}
\frac{1}{m} \sum_{i=1}^m \big \| \P_{\bOmega_i |W_C} \P_{\bR_\mi | \bX_i = \dummy, \bY_i = \dummy, W_C} - \P_{\bOmega_i | W_C} \P_{\bR_\mi |\bOmega_i  W_C}  \big \| \,=\, O \Big( \frac{\sqrt{\delta}}{\alpha^2} \Big)\;,
\end{equation}
which establishes the third item.
\item We insert a fresh copy of $\P_{\bX_i \bY_i | \bOmega_i}$ on both sides of the difference in each term of~\eqref{eq:classical-skew-2d} and~\eqref{eq:classical-skew-3} to get, using Lemma~\ref{lem:trivial},
\begin{gather*}
\frac{1}{m} \sum_{i=1}^m \big \| \P_{\bOmega_i |W_C}\P_{\bX_i \bY_i | \bOmega_i}  \P_{\bR_\mi | \bX_i, \bY_i, W_C} - \P_{\bOmega_i | W_C} \P_{\bX_i \bY_i | \bOmega_i} \P_{\bR_\mi |\bOmega_i  W_C}  \big \| \leq O(\sqrt{\delta})~,  \\
\frac{1}{m} \sum_{i=1}^m \big \| \P_{\bOmega_i |W_C}\P_{\bX_i \bY_i | \bOmega_i}  \P_{\bR_\mi | \bX_i = \dummy, \bY_i = \dummy, W_C} - \P_{\bOmega_i | W_C} \P_{\bX_i \bY_i | \bOmega_i} \P_{\bR_\mi |\bOmega_i  W_C}  \big \| \leq O \Big( \frac{\sqrt{\delta}}{\alpha^2} \Big)~.
\end{gather*}
Using the triangle inequality with the two preceding inequalities we get
\[
\frac{1}{m} \sum_{i=1}^m \big \| \P_{\bOmega_i |W_C}\P_{\bX_i \bY_i | \bOmega_i}  \P_{\bR_\mi | \bX_i = \dummy, \bY_i = \dummy, W_C} - \P_{\bOmega_i | W_C} \P_{\bX_i \bY_i | \bOmega_i} \P_{\bR_\mi | \bX_i \bY_i  W_C}  \big \| \leq O \Big( \frac{\sqrt{\delta}}{\alpha^2} \Big)~.
\]
Using~\eqref{eq:classical-skew-2bb} combined with \Cref{lem:trivial} and the triangle inequality we get
\[
\frac{1}{m} \sum_{i=1}^m \big \| \P_{\bOmega_i |W_C}\P_{\bX_i \bY_i | \bOmega_i}  \P_{\bR_\mi | \bX_i = \dummy, \bY_i = \dummy, W_C} - \P_{\bOmega_i} \P_{\bX_i \bY_i | \bOmega_i} \P_{\bR_\mi | \bX_i \bY_i  W_C}  \big \| \leq O \Big( \frac{\sqrt{\delta}}{\alpha^2} \Big)~.
\]
Marginalizing over the random variable $\bOmega_i$, the preceding equation becomes
\[
\frac{1}{m} \sum_{i=1}^m \Big \| \P_{\bX_i \bY_i}  \P_{\bR_\mi | \bX_i = \dummy, \bY_i = \dummy, W_C} - \P_{\bX_i \bY_i} \P_{\bR_\mi | \bX_i \bY_i  W_C}  \Big \| \leq O \Big( \frac{\sqrt{\delta}}{\alpha^2} \Big)~,
\]
which concludes the proof of Item 4.
\end{enumerate}
\end{proof}

\subsection{Quantum states and operators}
\label{sec:states-operators} 

From the operators and state specified in the strategy $\strategy^n$ we define new operators and states that will be used in the analysis.

\paragraph{Operators.} For all $\bx \in \X^n$ and $\by \in \Y^n$ let $\{A_\bx(\ba)\}_{\ba \in \A^n}$ and $\{B_\by(\bb)\}_{\bb \in \B^n}$ denote the players' POVMs in the strategy $\strategy^n$ when receiving questions $\bx$ and $\by$, respectively. Define the following operators for each $\ba_C \in \A^C, \bb_C \in \B^C, \bx \in \X^n, \by \in \Y^n$.
\begin{equation}
\label{eq:states-operators-1}
	A_{\bx}(\ba_C) = \sum_{\ba | \ba_C} A_{\bx}(\ba) \qquad\text{and} \qquad B_{\by}(\bb_C) = \sum_{\bb | \bb_C} B_{\by}(\bb)\;,
\end{equation}
where $\ba | \ba_C$ (resp. $\bb | \bb_C$) indicates summing over all tuples $\ba$ consistent with $\ba_C$ (resp. $\bb$ consistent with $\bb_C$). Note that for all $\bx$ (resp. $\by$), the set $\{ A_\bx(\ba_C) \}$ (resp. $\{ B_\by(\bb_C) \}$) denotes a POVM with outcomes in the set $\A^C$ (resp. $\B^C$).

For all $i \in [m]$, $\bomega_\mi$, $x \in \X$, and $y \in \Y$ define
\begin{equation}
\label{eq:states-operators-2}
 A_{\bomega_\mi, x}(\ba_C) = \Ex_{\bX | \bomega_\mi, x} A_{\bx}(\ba_C) \qquad \text{and} \qquad  B_{\bomega_\mi, y}(\bb_C) = \Ex_{\bY | \bomega_\mi, y} B_{\by}(\bb_C)\;,
\end{equation}
where $\Ex_{\bX | \bomega_\mi, x}$ is shorthand for $\Ex_{\bX | \bOmega_\mi = \bomega_\mi, \bX_i = x}$ and similarly for $\Ex_{\bY | \bomega_\mi, y}$.
Let $\X_{/\dummy} = \{ \dummyx : x\in \X\}$ be a disjoint copy of $\X$. Here, for each $x\in \X$, ``$\dummyx$'' is a new symbol that is used to distinguish elements in $\X$ from elements in $\X_{/\dummy}$. 
For all $\dummyx \in \X_{/\perp}$ define 
\begin{equation}
\label{eq:states-operators-3}
	A_{\bomega_\mi,\dummyx}(\ba_C) \,=\, \eta \, A_{\bomega_\mi,\dummy}(\ba_C) + (1 - \eta) \, A_{\bomega_\mi,x}(\ba_C)\;.
\end{equation}
Note that $A_{\bomega_\mi,\dummyx}(\ba_C)$ can be equivalently defined as $\Ex_{\bX | \bOmega_\mi = \bomega_\mi, \bM_i = x} A_{\bx}(\ba_C)$. Using that all operators are positive semidefinite we observe for later use that
\begin{align}
	 A_{\bomega_\mi,\dummy}(\ba_C) &\leq\, \frac{1}{\eta}\, A_{\bomega_\mi,\dummyx}(\ba_C)\;,\label{eq:states-operators-3a}\\
 A_{\bomega_\mi,x}(\ba_C) &\leq\, \frac{1}{1-\eta}\, A_{\bomega_\mi,\dummyx}(\ba_C)\;.\label{eq:states-operators-3b}
\end{align}

\paragraph{States.}  For all $i \in [n] \setminus C$, $\br_\mi = (\bomega_\mi,\ba_C,\bb_C)$ and $x \in \X$, for all $s \in \{x,\dummyx\}$, and for all $y \in \Y$, define the (unnormalized) state
\begin{equation}\label{eq:def-state-y}
	\ket{\Phi_{\br_\mi,s,y}} \,=\, A_{\bomega_\mi,s}(\ba_C)^{1/2} \otimes B_{\bomega_\mi,y}(\bb_C)^{1/2} \, \ket{\psi}
\end{equation}
and the normalization factor
\begin{equation}\label{eq:def-gamma}
	\gamma_{\br_\mi,s,y} \,=\, \big \| \, \ket{\Phi_{\br_\mi,s,y}} \, \big \|~.
\end{equation}
Finally we let
\begin{equation}\label{eq:def-psitt}
	\ket{\wt{\Phi}_{\br_\mi,s,y}} = \gamma_{\br_\mi,s,y}^{-1} \, \ket{\Phi_{\br_\mi,s,y}}
\end{equation}
denote the normalized version of the state $\ket{\Phi_{\br_\mi,s,y}}$.

For notational convenience we often suppress the dependence on $i$, $\bomega_\mi$, $\ba_C$, and $\bb_C$ when it is clear from context. Thus, for example, when we refer to an operator such as $A_\dummyx$, we really mean the operator $A_{\bomega_\mi,\dummy\!/\bx_i}(\ba_C)$ where $\bx_i = x$. As another example, we will often write $\ket{\wt{\Phi}_{x,y}}$ to denote the state $\ket{\wt{\Phi}_\ssxy}$.

The following proposition expresses the normalization factors $\gamma_{\br_\mi,s,y}$ as a function of the probability of obtaining answers $(\ba_C,\bb_C)$ in the strategy $\strategy^n$. 

\begin{proposition}
\label{prop:gamma}
For all $s \in \X$, 
\begin{equation}\label{eq:gamma-def-1}
	\gamma_{\br_\mi,s,y} = \Big( \P_{\bA_C \bB_C | \bOmega_\mi = \bomega_\mi, \bX_i = s, \bY_i = y}(\ba_C,\bb_C) \Big)^{1/2}\;,
\end{equation}
and for all $s = \dummyx \in \X_{/\perp}$,
\begin{align}\label{eq:gamma-def-2}
	\gamma_{\br_\mi,s,y} &= \Big( \eta \, \P_{\bA_C \bB_C | \bOmega_\mi = \bomega_\mi, \bX_i = \dummy, \bY_i = y}(\ba_C,\bb_C) + (1 - \eta) \, \P_{\bA_C \bB_C | \bOmega_\mi = \bomega_\mi, \bX_i = x, \bY_i = y}(\ba_C,\bb_C)\Big)^{1/2} \\
	&= \P_{\bA_C \bB_C | \bOmega_\mi = \bomega_\mi, \bOmega_i = (A,x), \bY_i = y}(\ba_C,\bb_C)^{1/2}~.\label{eq:gamma-def-3}
\end{align}
\end{proposition}

\begin{proof}
Suppose first that $s \in \X$. Expanding the definition of $ A_{\bomega_\mi,s}(\ba_C)$ and $B_{\bomega_\mi,y}(\bb_C)$,
\begin{align*}
	\gamma_{\br_\mi,s,y}^2 &= \bra{\psi} A_{\bomega_\mi,s}(\ba_C) \otimes B_{\bomega_\mi,y}(\bb_C) \, \ket{\psi} \\
	&= \Ex_{\bX | \bomega_\mi, s} \, \Ex_{\bY | \bomega_\mi, y} \, \bra{\psi} A_\bx(\ba_C) \otimes B_\by(\bb_C) \, \ket{\psi} \\
	&= \Ex_{\bX | \bomega_\mi, s} \, \Ex_{\bY | \bomega_\mi, y} \, \bra{\psi} A_\bx(\ba_C) \otimes B_\by(\bb_C) \, \ket{\psi} \\
	&= \Ex_{\bX \bY | \bOmega_\mi = \bomega_\mi, \bX_i = s, \bY_i = y} \, \P_{\bA_C \bB_C | \bx, \by}(\ba_C,\bb_C) \\
	&= \P_{\bA_C \bB_C | \bOmega_\mi = \bomega_\mi, \bX_i = s,\bY_i = y}(\ba_C,\bb_C)~.
\end{align*}
In the second-to-last line, we used the fact that $\bX$ and $\bY$ are independent conditioned on $\bOmega_\mi, \bX_i, \bY_i$. 
The calculation for when $s = \dummyx$ follows in nearly an identical manner to establish~\eqref{eq:gamma-def-2}. The equality in~\eqref{eq:gamma-def-3} follows from how $\bX_i$ and $\bOmega_i$ are coupled together (see \Cref{sec:p_setup}). 
\end{proof}

 \section{Existence of local unitaries}
\label{sec:unitaries}

This section uses the same setup as Section~\ref{sec:quantum_setup}, with $G = (\X \times \Y,\A \times \B,\mu,V)$ an $\alpha$-anchored two-player game, $n\geq 1$ and $\strategy^n = (\ket{\psi},A,B)$ a strategy for $G$.  
The main result of the section, and the only result used outside of it, is the following proposition. The proposition specifies the existence of unitaries $U_{\br_\mi,x}$ and $V_{\br_\mi,y}$ used  in the proof of Theorem~\ref{thm:anchorpr_quantum} to define a strategy $\strategy$ for the game $G$ from the strategy $\strategy^n$. 
Recall the definition of the states $\ket{\wt{\Phi}_{\br_\mi,s,y}}$, for $s\in \X\cup \X_{/\dummy}$ and $y\in \Y$, in~\eqref{eq:def-psitt}.

\begin{proposition}
\label{prop:local_unitaries}
	For every $C\subseteq [n]$, $i\in[n]\backslash C$, $\br_\mi$, $x$ and $y$, there exists unitaries $U_{\br_\mi,x}$ acting on $E_A$ and $V_{\br_\mi,y}$ acting on $E_B$ such that
	$$
		\Ex_I \,\, \Ex_{\bR_\mi | W_C} \,\, \Ex_{XY}  \big \| (U_{\br_\mi,x} \otimes V_{\br_\mi,y}) \ket{\wt{\Phi}_{\sspp}} - \ket{\wt{\Phi}_{\ssxy}} \big\| \,=\, O\big(\delta^{1/16}/\alpha^3\big)\;,
	$$
	where $\Ex_I$ denotes the expectation over a uniformly random $i \in [n] \setminus C$, $\Ex_{\bR_\mi | W_C}$ denotes the expectation over $\br_\mi$ sampled from $\P_{\bR_\mi | W_C}$, and $\Ex_{XY}$ denotes the expectation over $(x,y)$ sampled from $\mu$. 
\end{proposition}

The remainder of the section is devoted to the proof of Proposition~\ref{prop:local_unitaries}. The proof is based on two lemmas. The first, given in Section~\ref{sec:l1}, defines the unitaries $U_{\br_\mi,x}$ and $V_{\br_\mi,y}$ as well as additional unitaries $V_{\br_\mi,x,y}$ that are used in the proof of Proposition~\ref{prop:local_unitaries}. The second, given in Section~\ref{sec:l2}, relates the normalization factors $\gamma_\xy$ and $\gamma_{\pxy}$ defined in~\eqref{eq:def-gamma}. In the next subsection we first present some results from quantum information theory. For a more comprehensive reference we refer the reader to the book~\cite{wilde2013quantum}.
 
\subsection{Some tools from quantum information theory}

For a Hermitian operator $X$ let $\mathrm{supp}(X)$ denote the projection onto its image. The \emph{relative entropy} between two positive semidefinite operators $\rho$, $\sigma$, denoted by $\D(\rho \| \sigma)$, is defined to be $\Tr(\rho (\log \rho - \log \sigma))$. The \emph{relative min-entropy} $\D_\infty(\rho \| \sigma)$ is defined as $\min\{ \lambda : \rho \preceq 2^\lambda \sigma \}$, and $+\infty$ if this set is empty. Let $\rho^{AB}$ be a bipartite state. The \emph{mutual information} $\I(A:B)_\rho$ between registers $A$ and $B$ in state $\rho$ is defined as $\D(\rho^{AB} \| \rho^A \otimes \rho^B)$. We define the \emph{conditional mutual information} $\I(A:B|C)_\rho$ for a tripartite quantum state $\rho^{ABC}$ as $\I(A:BC)_\rho - \I(A:C)_\rho$. When $\rho$ is classical on $C$, then $\I(A:B|C)_\rho = \Ex_{x \sim \rho^C} \I(A:B)_{\rho_x}$ where $x \sim \rho^C$ denotes sampling $x$ from the classical distribution $\rho^C$.

\begin{proposition}[Theorem 11.8.1 of~\cite{wilde2013quantum}]
\label{prop:relative_entropy_nonneg}
For all density matrices $\rho, \sigma$, the relative entropy $\D(\rho \| \sigma)$ is nonnegative.
\end{proposition}

\begin{proposition}[Strong subadditivity of conditional mutual information, Theorem 11.7.1 of~\cite{wilde2013quantum}]
\label{prop:ssa}
For all tripartite density matrices $\rho^{ABC}$, the conditional mutual information $\I(A:B|C)_\rho$ is nonnegative.
\end{proposition}

\begin{proposition}[Pinsker's inequality, Theorem 11.9.2 of~\cite{wilde2013quantum}]
\label{prop:pinsker}
	For all density matrices $\rho, \sigma$, 
	\[
		 \frac{1}{2 \ln 2} \| \rho - \sigma \|^2_1 \leq \D(\rho \| \sigma)~.
	\]
\end{proposition}

\begin{proposition}[Theorem 11.9.1 in~\cite{wilde2013quantum}]
\label{prop:divergence_data_processing}
Let $\rho^{XY}$ and $\sigma^{XY}$ be quantum states. Then $\D(\rho^{X} \| \sigma^X) \leq \D(\rho^{XY} \| \sigma^{XY})$.
\end{proposition}

\begin{proposition}[Chain rule for relative entropy]
\label{prop:divergence_chain_rule}
	Let $\rho = \sum_x \P(x) \ketbra{x}{x} \otimes \rho_x$ and $\sigma = \sum_x \Qsf(x) \ketbra{x}{x} \otimes \sigma_x$ for some probability distributions $\P$, $\Qsf$. Then $\D(\rho \| \sigma) = \D(\P \| \Qsf) + \Ex_{x \sim \P} \D(\rho_x \| \sigma_x) $, where $\D(\P \| \Qsf) = \sum_x \P(x) \log \frac{\P(x)}{\Qsf(x)}$ denotes the relative entropy between two distributions. In particular, $\D(\rho \| \sigma) \geq \Ex_{x \sim \P} \D(\rho_x \| \sigma_x)$.
\end{proposition}

\begin{proof}
	For all $x$, let $\Pi_x = \mathrm{supp}(\rho_x)$ and $\Gamma_x = \mathrm{supp}(\sigma_x)$. We note that 
	\[
	\log \rho = \sum_x \log \P(x) \ketbra{x}{x} \otimes \Pi_x + \sum_x \ketbra{x}{x} \otimes \log \rho_x~,
	\]
	 and therefore 
	 \[
	 	\rho \log \rho = \sum_x \P(x) \log \P(x) \ketbra{x}{x} \otimes \rho_x + \sum_x \P(x) \ketbra{x}{x} \otimes \rho_x \log \rho_x~.
	\]
	Similarly, 
	\[
		\rho \log \sigma = \sum_x \P(x) \log \Qsf(x) \ketbra{x}{x} \otimes \rho_x \Gamma_x + \sum_x \P(x) \ketbra{x}{x} \otimes \rho_x \log \sigma_x~.
	\] 
	Subtracting, we get
	\begin{align}
		\D(\rho \| \sigma) &= \sum_x \P(x) \Tr \Big( \ketbra{x}{x} \otimes \Big( \log \P(x) \rho_x - \log \Qsf(x) \rho_x \Gamma_x + \rho_x (\log \rho_x - \log \sigma_x) \Big)  \Big ) \notag \\
		&= \sum_x \P(x) \Big( \log \P(x) \, \Tr(\rho_x) - \log \Qsf(x) \, \Tr( \rho_x \Gamma_x ) +  \Tr \Big( \rho_x (\log \rho_x - \log \sigma_x) \Big) \Big) \notag \\
		&= \sum_x \P(x) \Big( \log \P(x)  - \log \Qsf(x) \, \Tr( \rho_x \Gamma_x ) \Big) + \P(x) \, \D(\rho_x \| \sigma_x) \label{eq:divergence_chain_rule}
	\end{align}
	Suppose that for some $x$, $\rho_x \Gamma_x \neq \rho_x$. This implies that the support of $\sigma_x$ does not contain the support of $\rho_x$, and therefore $\D(\rho_x \| \sigma_x) = \infty$. Thus $\D(\P \| \Qsf) + \Ex_{x \sim \P} \D(\rho_x \| \sigma_x) = \infty$ and $\D(\rho \| \sigma) = \infty$ (because each term of the sum in~\eqref{eq:divergence_chain_rule} is nonnegative), establishing equality. Otherwise, $\rho_x \Gamma_x = \rho_x$ for all $x$, and thus~\eqref{eq:divergence_chain_rule} is again equal to $\D(\P \| \Qsf) + \Ex_{x \sim \P} \D(\rho_x \| \sigma_x)$. 
	
	The ``in particular'' statement follows from the fact that $\D(\P \| \Qsf)$ is nonnegative. 
\end{proof}

\begin{proposition}
\label{prop:mutual-information-classical}
Let $\rho^{XE} = \sum_x \P(x) \ketbra{x}{x}^X \otimes \rho_x^E$ denote a classical-quantum state. Then 
\[\I(X:E)_\rho \,=\, \Ex_{x \sim \P} \D(\rho_x^E \| \rho^E)\;.\]
\end{proposition}

\begin{proof}
	By definition of the mutual information, $\I(X : E)_\rho = \D(\rho^{XE} \| \rho^X \otimes \rho^E)$. Letting 
	\[\sigma = \sum_x \P(x) \ketbra{x}{x}^X \otimes \rho^E\]
	we can apply \Cref{prop:divergence_chain_rule} to get
	\[
		\I(X : E)_\rho = \D(\rho^{XE} \| \rho^X \otimes \rho^E) = \Ex_{x \sim \P} \D(\rho_x^E \| \rho^E)~.
	\]
\end{proof}

\begin{proposition}
\label{prop:relative_min_entropy_chain_rule2}
Let $\rho$, $\sigma$, and $\tau$ be density matrices such that $\D(\rho \| \sigma) \leq \lambda_1$ and $\D_\infty(\sigma \| \tau) \leq \lambda_2$. Then $\D_\infty(\rho \| \tau) \leq \lambda_1 + \lambda_2$.
\end{proposition}
\begin{proof}
	$\D_\infty(\sigma \| \tau) = \lambda_2$ implies that $2^{-\lambda_2} \sigma \preceq \tau$. Then,
	\begin{align*}
		\D(\rho \| \tau) &= \Tr(\rho (\log \rho - \log \tau)) \leq \Tr(\rho(\log \rho - \log 2^{-\lambda_2} \sigma)) \\
		&\leq \Tr(\rho(\log \rho - (-\lambda_2)\Id - \log \sigma)) \\
		&\leq \lambda_2 + \Tr(\rho (\log \rho - \log \sigma)) = \lambda_1 + \lambda_2\;.
	\end{align*}
\end{proof}

\begin{proposition}[Quantum Gibbs's Inequality]
\label{prop:divergence_gibbs_inequality}
	For quantum states $\rho^{AB}$, $\sigma^A$, $\tau^B$, 
	\[\D(\rho^{AB} \big \| \sigma^A \otimes \tau^B) \,\geq\, \I(A : B)_\rho\;.\]
\end{proposition}

\begin{proof}
	If the images of $\rho^{A}$ and $\rho^B$ are not contained within the images of $\sigma^A$ and $\tau^B$ respectively then $\D(\rho^{AB} \big \| \sigma^A \otimes \tau^B) = \infty$ and the inequality trivially holds. Otherwise, using the fact that $\log (\sigma \otimes \tau) = (\log \sigma)\otimes \Gamma + \Pi \otimes (\log \tau)$ we have that
	\begin{align*}
	&\D(\rho^{AB} \big \| \sigma^A \otimes \tau^B) - \I(A:B)_\rho \\
	&= \D(\rho^{AB} \big \| \sigma^A \otimes \tau^B) - \D(\rho^{AB} \big \| \rho^A \otimes \rho^B) \\
		&=  \Tr \big( \rho^{AB} \big( (\log \rho^A) \otimes \mathrm{supp}(\rho^B) - (\log \sigma^A) \otimes \mathrm{supp}(\tau^B) \big) \big) \\
		& \qquad \qquad + \Tr \big( \rho^{AB} \big( \mathrm{supp}(\rho^A) \otimes (\log \rho^B)  - \mathrm{supp}(\sigma^A) \otimes (\log \tau^B) \big) \big) \\
		&= \D(\rho^A \| \sigma^A) + \D(\rho^B \| \tau^B) \geq 0\;,
	\end{align*}
	where in the last line we used that $\rho^{AB} (\Id \otimes \mathrm{supp}(\tau^B)) = \rho^{AB} (\mathrm{supp}(\sigma^A) \otimes \Id) = \rho^{AB}$ and that the relative entropy is nonnegative (\Cref{prop:relative_entropy_nonneg}). 
\end{proof}

We prove a quantum analogue of \emph{Raz's Lemma}, a central tool used in many proofs of parallel repetition theorems~\cite{raz1998parallel, Hol09,barak2009strong}. 

\begin{lemma}[Quantum Raz's Lemma]
\label{lem:quantum_raz}
Let $\rho$ and $\sigma$ be two CQ states with  $\rho^{XA}= \rho^{X_1 X_2 \ldots X_n A}$ and $\sigma= \sigma^{XA}= \sigma^{X_1}\otimes \sigma^{X_2}\otimes \ldots \otimes \sigma^{X_n} \otimes \sigma^A$ with $X=X_1 X_2 \ldots X_n$ classical in both states. Then
\begin{equation}\label{eqn:Raz_lemma1} 
\sum_{i=1}^n \,\I(X_i \, :\, A)_\rho \,\leq\, \D(\rho^{XA} \, \| \sigma^{XA})\;. 
\end{equation}
\end{lemma}

\begin{proof}
By the chain rule (\Cref{prop:divergence_chain_rule}) we have 
\begin{equation}\label{eqn:Raz_lemma2}
\D(\rho^{XA} \| \sigma^{XA})= \D(\rho^{X_1} \| \sigma^{X_1}) + \Ex_{x_1 \sim \rho^{X_1}} \D(\rho^{X_2}_{X_1=x_1}\| \sigma^{X_2}) + \ldots+ \Ex_{x \sim \rho^{X_1 \cdots X_n}} \D(\rho^{A}_{X=x} \| \sigma^{A})\;,
\end{equation}
where $x_1 \sim \rho^{X_1}$ means sampling $x_1$ according to the classical distribution $\rho^{X_1}$, and similarly for $x \sim \rho^{X_1 \cdots X_n}$. Consider the $i$-th term in~\eqref{eqn:Raz_lemma2}, for $i\in\{1,\ldots,n\}$. We have, via the chain rule followed by Quantum Gibbs's Inequality (\Cref{prop:divergence_gibbs_inequality}), 
\begin{equation}\label{eq:raz-1a}
\Ex_{x_{<i}\sim \rho^{X_1X_2\ldots X_{i-1}}} \D(\rho^{X_i}_{X_{< i}= x_{<i}} \| \sigma^{X_i}) = \D(\rho^{X_1 \cdots X_i} \|\rho^{X_1 \cdots X_{i-1}} \otimes \sigma^{X_i}) \geq  \I(X_1\ldots X_{i-1} : X_i)_\rho\;.
\end{equation}
Now consider the last term in~(\ref{eqn:Raz_lemma2}). Again by the Quantum Gibbs's Inequality we have
\begin{align}
\Ex_{x\sim \rho^X} \D(\rho^{A}_{X=x} \|\sigma^{A}) &= \D(\rho^{XA} \| \rho^X \otimes \sigma^A) \geq \I(X:A)_\rho = \sum_{i=1}^n \I(X_i : A| X_1 X_2 \ldots X_{i-1})_\rho\;,\label{eq:raz-1b}
\end{align}
where the last equality is by definition of the conditional mutual information. 
Summing~\eqref{eq:raz-1a} for $i\in\{1,\ldots,n\}$ and~\eqref{eq:raz-1b} and using $\I(X_i:A X_1\ldots X_i)= \I(X_i:X_1\ldots X_{i-1})+ \I(X_i:A|X_1\ldots X_{i-1})$ gives, from~\eqref{eqn:Raz_lemma2},
\[ \D(\rho^{X A} \| \sigma^{XA})\geq \sum_{i=1}^n \I(X_i: A X_1 \ldots X_{i-1})_\rho \geq \sum_{i=1}^n \I(X_i:A)_\rho\;,
\]
where the last inequality follows by taking the difference and using strong subadditivity of conditional mutual information (\Cref{prop:ssa}). 
\end{proof}

Define the \emph{fidelity} between two density matrices $\rho$ and $\sigma$ as 
\[\F(\rho,\sigma) = \| \sqrt{\rho} \sqrt{\sigma} \|_1\;.\]
The Fuchs-van de Graaf inequalities relate fidelity and trace norm as
\begin{equation}\label{eq:fuchs-graaf}
1 - \F(\rho,\sigma) \leq \frac{1}{2} \| \rho - \sigma \|_1 \leq \sqrt{1 - \F(\rho,\sigma)^2}\;.
\end{equation}
The following relates the fidelity between two density matrices with the inner product between their purifications. 

\begin{theorem}[Uhlmann's Theorem, Theorem 9.2.1 in~\cite{wilde2013quantum}]
\label{thm:uhlmann}
	Let $\ket{\psi}^{AB},\ket{\phi}^{AB}$ be bipartite states, and let $\rho^{A}$ and $\sigma^{A}$ denote their reduced density matrices on the system $A$, respectively. Then there exists a unitary map $V$ acting on $B$ such that
	\[
		\bra{\phi} (\Id \otimes V) \ket{\psi} = \F(\rho,\sigma)~.
	\]
\end{theorem}

\subsection{First lemma}
\label{sec:l1}

The first lemma is the following. 

\begin{lemma}
\label{lem:unitary_bounds}
For all $i$, $\br_\mi$, $x$ and $y$ there exists unitaries $U_{\br_\mi, x}$ acting on $E_A$ and unitaries $V_{\br_\mi, y}$, $V_{\ssxy}$ acting on $E_B$ such that with probability at least $1 - O(\delta^{1/16})$ over the choice of a uniformly random $i \in [n] \setminus C$, 
\begin{align}
	\Ex_{\bR_\mi | W_C} \,\, \Ex_X \, \, \, \big \|(U_{\br_\mi, x} \otimes \Id)  \ket{\wt{\Phi}_\sspp} - \ket{\wt{\Phi}_\ssxp }   \big \| &=  O(\delta^{1/16}/\alpha^{5/4})\;,\label{eq:ux_bound}\\
	\Ex_{\bR_\mi | W_C} \,\, \Ex_Y \, \, \, \big \| (\Id \otimes V_{\br_\mi,y})  \ket{\wt{\Phi}_\sspp}  - \ket{\wt{\Phi}_\sspy }  \big \| &=  O(\delta^{1/16}/\alpha^{5/4})\;,\label{eq:vy_bound}\\
	\Ex_{\bR_\mi | W_C} \,\, \Ex_{XY} \,\, \, \big \|  (\Id \otimes V_{\ssxy})  \ket{\wt{\Phi}_\sspxy}  - \ket{\wt{\Phi}_\sspxp }  \big \| &=  O(\delta^{1/16}/\alpha^{5/4})\;.\label{eq:vxy_bound}
		\end{align}
		where $\Ex_X$, $\Ex_Y$, and $\Ex_{XY}$ denote expectations under $\mu_X(x)$, $\mu_Y(y)$, and $\mu_{XY}(x,y)$ respectively.
\end{lemma}

The remainder of this subsection is devoted to the proof of Lemma~\ref{lem:unitary_bounds}. 
Recall that the strategy $\strategy^n$ consists of the entangled state $\ket{\psi} \in \C^d_{E_A} \otimes \C^d_{E_B}$ and POVMs $\{A_\bx(\ba) \}$ acting on system $E_A$ and $\{B_\by(\bb) \}$ acting on system $E_B$. For all $\bomega$, $\ba_C$, and $\bb_C$ define
\begin{equation}
\label{eq:a-bomega-def}
	A_\bomega(\ba_C) = \Ex_{\bX | \bOmega = \bomega} \, A_\bx(\ba_C) \qquad \text{and} \qquad B_\bomega(\bb_C) = \Ex_{\bY | \bOmega = \bomega} \, B_\by(\bb_C)
\end{equation}
where $A_\bx(\ba_C),B_\by(\bb_C)$ are defined in~\eqref{eq:states-operators-1}. 
We let $\rho$ denote the reduced density matrix of $\ket{\psi}$ on either system (this is well-defined because we assumed the form~\eqref{eq:psi-sym} for $\ket{\psi}$). 

Introduce the notation $\psi = \ket{\psi}\bra{\psi}$ and $X[\rho]=X\rho X^\dagger$. For a classical random variable, such as $\ba_C$, we write $\puretomixed{\ba_C}$ to denote the rank-$1$ density matrix $\ketbra{\ba_C}{\ba_C}$. Let $\ac$ denote pair of the random variables $(\bA_C, \bB_C)$ so that $\bR_\mi = (\bOmega_\mi,\ac)$. Define the following density operators.
\begin{gather}
\Xi^{\bOmega \bY E_AE_B\ac} = \sum_{\bomega, \by, \ba_C, \bb_C} \P_{\bOmega \bY} (\bomega, \by) \, \puretomixed{\bomega, \by}   \otimes \left (\sqrt{A_{\bomega}(\ba_C)} \otimes \sqrt{B_{\by}(\bb_C)} \right ) \left [\psi \right ] \otimes \puretomixed{\ba_C, \bb_C}~, \label{eq:Xi-def} \\
\Lambda^{\bOmega \bX E_AE_B\ac} = \sum_{\bomega, \bx, \ba_C, \bb_C} \P_{\bOmega \bX} (\bomega, \bx) \, \puretomixed{\bomega, \bx}   \otimes \left (\sqrt{A_{\bx}(\ba_C)} \otimes \sqrt{B_{\bomega}(\bb_C)} \right ) \left [\psi \right ] \otimes \puretomixed{\ba_C, \bb_C}~. \label{eq:Lambda-def}
\end{gather}
Both states are classical on the registers $\bOmega$, $\bX$, $\bY$, and $\ac$ and quantum on registers $E_A$ and $E_B$. The state $\Xi$ is defined such that when tracing out the entanglement registers $E_A$ and $E_B$ the resulting state $\Xi^{\bOmega \bY \bA_C \bB_C}$ is a classical state representing the probability distribution $\P_{\bOmega \bY \bA_C \bB_C}$. To see this, observe that
\begin{align*}
\Xi^{\bOmega \bY \ac} &= \sum_{\bomega, \by, \ba_C, \bb_C} \P_{\bOmega \bY}(\bomega, \by) \, \puretomixed{\bomega, \by, \ba_C, \bb_C} \, \bra{\psi} A_\bomega(\ba_C) \otimes B_\by(\bb_C) \ket{\psi} \\
&= \sum_{\bomega, \by, \ba_C, \bb_C} \P_{\bOmega \bY}(\bomega, \by) \, \puretomixed{\bomega, \by, \ba_C, \bb_C} \, \bra{\psi} \Ex_{\bX | \bOmega = \bomega} \, A_\bx(\ba_C)  \otimes B_\by(\bb_C) \ket{\psi} \\
&= \sum_{\bomega, \bx, \by, \ba_C, \bb_C} \P_{\bOmega \bX \bY}(\bomega, \bx, \by) \, \puretomixed{\bomega, \by, \ba_C, \bb_C} \cdot \P_{\bA_C \bB_C | \bX = \bx, \bY = \by}(\ba_C,\bb_C) \\
&= \sum_{\bomega,\by, \ba_C, \bb_C} \P_{\bOmega\bY \bA_C \bB_C}(\bomega, \by, \ba_C,\bb_C) \, \puretomixed{\bomega, \by, \ba_C, \bb_C}\;,
\end{align*}
where in the third line we used Item 1 from Claim~\ref{claim:ufacts} and in the fourth line we used Item 2 from Claim~\ref{claim:ufacts}. Similarly, the state $\Lambda^{\bOmega \bX \bA_C \bB_C}$ represents the probability distribution  $\P_{\bOmega \bX \bA_C \bB_C}$. In particular, we obtained that $\Xi^{\bOmega \bY E_AE_B\ac} $ and $\Lambda^{\bOmega \bY E_AE_B\ac} $ are properly normalized density matrices. 

Since the event $W_C$ is determined by the random variables $(\bOmega,\bA_C,\bB_C)$ we can {condition} the states $\Xi, \Lambda$ on the event $W_C$ to obtain states
\begin{gather*}
	\xi^{\bOmega \bY E_A E_B \ac} = \frac{1}{\P(W_C)} \sum_{\substack{\bomega, \by, \ba_C, \bb_C: \\ (\bomega,\ba_C,\bb_C) \in W_C}} \P_{\bOmega \bY}(\bomega, \by) \, \puretomixed{\bomega, \by} \otimes \left (\sqrt{A_{\bomega}(\ba_C)} \otimes \sqrt{B_{\by}(\bb_C)} \right ) \left [\psi \right ] \otimes \puretomixed{\ba_C,\bb_C}~, \\
	\lambda^{\bOmega \bX E_A E_B \ac} = \frac{1}{\P(W_C)} \sum_{\substack{\bomega, \bx, \ba_C, \bb_C: \\ (\bomega,\ba_C,\bb_C) \in W_C}} \P_{\bOmega \bX}(\bomega, \bx) \, \puretomixed{\bomega, \bx} \otimes \left (\sqrt{A_{\bx}(\ba_C)} \otimes \sqrt{B_{\bomega}(\bb_C)} \right ) \left [\psi \right ] \otimes \puretomixed{\ba_C,\bb_C}\;.
\end{gather*}
We can further condition $\xi$ and $\Xi$ on specific settings of the classical variables $(\bOmega,\bY,\bA_C,\bB_C)$. For example, for any $\br = (\bomega,\ba_C,\bb_C)$ we write $\xi_\br$ and $\Xi_\br$ to denote $\xi$ and $\Xi$ conditioned on $\bOmega = \bomega, \bA_C = \ba_C,\bB_C = \bb_C$, respectively. As another example, for all $i \in [n] \setminus C$, $\br_\mi = (\bomega_\mi,\ba_C,\bb_C) \in W_C$, $x \in \X$ and $y \in \Y$, we write $\xi_{\br_\mi,x,y}$ and $\Xi_{\br_\mi,x,y}$ to denote $\xi$ and $\Xi$ respectively conditioned on $\bOmega = \bomega$ (where $\bomega = (\bomega_\mi,\bomega_i)$ with $\bomega_i = (A,x)$),  $\bA_C = \ba_C, \bB_C = \bb_C$, and $\bY_i = y$.

Similarly, let $\lambda_{\br_\mi,x,y}$ and $\Lambda_{\br_\mi,x,y}$ denote $\lambda$ and $\Lambda$ respectively conditioned on $\bOmega = \bomega$ (where $\bomega = (\bomega_\mi,\bomega_i)$ with $\bomega_i = (B,y)$), $\bA_C = \ba_C, \bB_C = \bb_C$, and $\bX_i = x$.

The main step of the proof is given by the following two claims which build on top of each other. 

\begin{claim}\label{claim:xi-change-x}
The following hold.
\begin{gather}
\Ex_I \,  \Ex_{\bR | W_C}\, \I \big (\bY_i ; E_A\big )_{\xi_\br} \,= \, O(\delta)~, \label{eq:xi-change-1}\\
\Ex_I \,  \Ex_{\bR | W_C} \,\I \big (\bX_i ; E_B\big )_{\lambda_\br} \,= \, O(\delta)~. \label{eq:xi-change-2}
\end{gather}
\end{claim}

\begin{proof}
We present the proof for~\eqref{eq:xi-change-1}; the proof for~\eqref{eq:xi-change-2} is similar. 
First we observe that by definition it holds that $\P(W_C) \, \xi\,\, \leq \,\Xi$. Using the definition of relative entropy and relative min-entropy this inequality implies
\[ \D(\xi \| \Xi)\,\leq\, \D_\infty(\xi \| \Xi)\, \leq\, \log \frac{1}{\P(W_C)}\;.\]
  Using the chain rule for the relative entropy (\Cref{prop:divergence_chain_rule}) and the fact that tracing out registers can only decrease the relative entropy (\Cref{prop:divergence_data_processing}), 
\begin{equation}\label{eq:claim23-1}
	\Ex_{\bR | W_C} \D \big (\xi_{\br}^{\bY E_A} \,\, \| \,\, \Xi_{\br}^{\bY E_A} \big ) \leq \log \frac{1}{\P(W_C)}\;.
\end{equation}
For a state of the form $\ket{\psi} = \sum \sqrt{\lambda_j} \ket{v_j}\ket{v_j} \in \C^d \otimes \C^d$ and for all linear operators $X,Y$ acting on $\C^d$ a simple calculation shows that
\[
	\Tr_{E_B} \big( (X \otimes Y)\ketbra{\psi}{\psi}(X \otimes Y)^\dagger \big) \,=\, X \, \sqrt{\rho} \, \overline{Y^\dagger Y} \, \sqrt{\rho} \, X^\dagger\;,
\]
where $\rho = \sum \lambda_j \ket{v_j}\bra{v_j}$ and the complex conjugate is taken with respect to the orthonormal basis $\{\ket{v_j}\}$.  
	Using this, for all $\bomega$
	\begin{align}
		\Xi_\bomega^{\bY E_A \ac} &= \sum_{\by, \ba_C, \bb_C} \P_{\bY | \bomega}(\by) \,\, \puretomixed{\by} \otimes \sqrt{A_\bomega(\ba_C)} \,\, \sqrt{\rho} \,\, \overline{B}_{\by}(\bb_C) \,\, \sqrt{\rho} \,\, \sqrt{A_\bomega(\ba_C)}  \otimes \puretomixed{\ba_C \bb_C}\notag \\
		&\leq  \sum_{\by, \ba_C, \bb_C} \P_{\bY | \bomega}(\by) \,\, \puretomixed{\by} \otimes \sqrt{A_\bomega(\ba_C)} \,\, \sqrt{\rho} \,\, \overline{B}_{\by}(\bb_C) \,\, \sqrt{\rho} \,\, \sqrt{A_\bomega(\ba_C)} \otimes \Id \notag\\
		&= \sum_{\by, \ba_C} \P_{\bY | \bomega}(\by) \,\, \puretomixed{\by} \otimes \sqrt{A_\bomega(\ba_C)} \,\,\rho\,\, \sqrt{A_\bomega(\ba_C)} \otimes \Id \notag \\
		&= \Xi_\bomega^{\bY} \otimes \Xi_\bomega^{E_A} \otimes \Id \;,\label{eq:claim23-2}
	\end{align}
where the second-to-last equality uses $\sum_{\bb_C} \overline{B}_{\by}(\bb_C) = \Id$ and in the last equality we used that the reduced density matrices
\[\Xi_\bomega^{\bY}\,=\,\sum_{\by} \P_{\bY | \bomega}(\by) \,\, \puretomixed{\by}\qquad\text{and}\qquad \Xi_\bomega^{E_A}\,=\,\sum_{\ba_C} \sqrt{A_\bomega(\ba_C)} \,\,\rho\,\, \sqrt{A_\bomega(\ba_C)}\;.\]
Taking the partial trace on both sides of~\eqref{eq:claim23-2} we get
\begin{align}
\label{eq:local-unitaries-2}
\Xi_\bomega^{\bY E_A} \,\leq\, (|\A|\cdot |\B|)^{|C|} \,\, \Xi_{\bomega}^{\bY} \otimes \Xi_{\bomega}^{E_A}\;.
\end{align}
	From~\eqref{eq:local-unitaries-2} and the definition of $\D_\infty$ it follows that 
	\begin{equation}\label{eq:claim23-2b}
	\D_\infty \big (\Xi_{\bomega}^{\bY E_A} \big \| \Xi_{\bomega}^{\bY} \otimes \Xi_{\bomega}^{E_A} \big ) \leq |C| \cdot \log |\A| |\B|\;.
	\end{equation}
Applying Lemma~\ref{lem:quantum_raz},
\begin{align*}
		\Ex_I \,  \Ex_{\bR | W_C} \I \big (\bY_i ; E_A\big )_{\xi_\br} &\leq 		\frac{1}{m} \Ex_{\bR | W_C} \D \big (\xi_{\br}^{\bY E_A} \big \| \Xi_{\br}^{\bY} \otimes \Xi_{\br}^{E_A} \big )\notag\\
		&\leq \frac{1}{m} \Ex_{\bOmega | W_C} \D \big (\xi_{\bomega}^{\bY E_A} \big \| \Xi_{\bomega}^{\bY} \otimes \Xi_{\bomega}^{E_A} \big )\notag \\
		&\leq 		\frac{1}{m} \Big(\Ex_{\bOmega | W_C} \D \big (\xi_{\bomega}^{\bY E_A} \big \| \Xi_{\bomega}^{\bY E_A} \big) + \Ex_{\bOmega | W_C} \D_\infty \big (\Xi_{\bomega}^{\bY E_A} \big \| \Xi_{\bomega}^{\bY} \otimes \Xi_{\bomega}^{E_A} \big) \, \Big)\notag\\
		&\leq \frac{1}{m}\Big(\log \frac{1}{\P(W)} +  |C| \cdot \log |\A| |\B|\Big) = \delta\;,
\end{align*}
where the second inequality follows from \Cref{prop:divergence_chain_rule}, the third from \Cref{prop:relative_min_entropy_chain_rule2} and the last uses~\eqref{eq:claim23-1} and~\eqref{eq:claim23-2b} as well as the definition of $\delta$ in~\eqref{eq:delta}. This establishes~\eqref{eq:xi-change-1} of the Claim. The proof for~\eqref{eq:xi-change-2} is similar.
\end{proof}

The second claim relates the reduced densities on $E_A$ (resp. $E_B$) of the states $\xi$ (resp. $\lambda$) associated with different choices of $i,\br_\mi,x,y$. 

\begin{claim}\label{claim:xi-change-y}
The following hold:
\begin{gather}
\Ex_I \, \Ex_{\bR_\mi |  W_C} \, \, \, \Ex_{XY} \,\, \big \| \xi^{E_A}_\ssxy - \xi^{E_A}_\ssxp \big \|_1^2\; = \; O\big(\sqrt{\delta}/\alpha^4 \big)~, \label{eq:E_A_1} \\
\Ex_I \, \Ex_{\bR_\mi |  W_C} \, \, \, \Ex_{XY} \,\, \big \| \lambda^{E_B}_\ssxy - \lambda^{E_B}_\sspy \big \|_1^2\; = \; O\big(\sqrt{\delta}/\alpha^4 \big)~, \label{eq:E_A_2}
\end{gather}
where the expectation over $XY$ is with respect to the distribution $\mu_{XY}$. 
\end{claim}

\begin{proof}
We show~\eqref{eq:E_A_1}; the proof of~\eqref{eq:E_A_2} is similar. Consider, for all $i$ and $\br$ that implies the event $W_C$, the state
\[
	\xi_{\br}^{\bY_i E_A} = \Ex_{\bY_i | \bR = \br, W_C} \, \puretomixed{\by_i}^{\bY_i} \otimes \xi_{\br,\by_i}^{E_A}
\]
and note that it is classical in the register $\bY_i$.
Applying Pinsker's inequality (Lemma~\ref{prop:pinsker}) and then using \Cref{prop:mutual-information-classical} we get
\begin{align}
\Ex_I \, \Ex_{\bR \bY_i | W_C} \,\, \big \| \xi^{E_A}_{\br,\by_i} -  \xi^{E_A}_{\br} \big \|_1^2 &\leq 2  \ln 2 \, \Ex_I \,\Ex_{\bR \bY_i | W_C} \,\,  \D \big(  \xi^{E_A}_{\br,\by_i} \, \big \|\, \xi^{E_A}_{\br} \big )\notag\\
&= 2\ln 2\, \Ex_I \, \Ex_{\bR | W_C} \,\,  \I (\bY_i ; E_A)_{\xi_\br} \notag\\
& = O(\delta)~,\label{eq:claim24-1a}
\end{align} 
where the last line is by~\eqref{eq:xi-change-1}. Using \Cref{lem:trivial} to insert a copy of $\P_{\bX_i \bY_i}$ on both sides of the difference of Item 3 of \Cref{lem:classical_skew} yields
\[
\Ex_I \left \|  \P_{\bX_i \bY_i} \P_{\bOmega_i | W_C} \P_{\bR_\mi | \bX_i = \dummy, \bY_i =\dummy,  W_C} - \P_{\bX_i \bY_i} \P_{\bOmega_i | W_C} \P_{\bR_\mi | \bOmega_i W_C} \right \| \leq O(\sqrt{\delta}/\alpha^2)~.
\]
Inserting $\P_{\Omega_i | W_C}$ into both sides of the difference of Item 4 of \Cref{lem:classical_skew} yields
\[
\Ex_I \left \|  \P_{\bX_i \bY_i} \P_{\bOmega_i | W_C} \P_{\bR_\mi | \bX_i = \dummy, \bY_i =\dummy,  W_C} - \P_{\bX_i \bY_i} \P_{\bOmega_i | W_C} \P_{\bR_\mi | \bX_i \bY_i  W_C} \right \| \leq O(\sqrt{\delta}/\alpha^2)~.
\]
Using the triangle inequality with the previous two bounds yields
\begin{equation}
\label{eq:xi-change-y-1b}
\Ex_I \left \|  \P_{\bX_i \bY_i} \P_{\bOmega_i | W_C} \P_{\bR_\mi | \bOmega_i  W_C} - \P_{\bX_i \bY_i}  \P_{\bOmega_i | W_C} \P_{\bR_\mi | \bX_i \bY_i  W_C} \right \| \leq O(\sqrt{\delta}/\alpha^2)~.
\end{equation}
Marginalizing over $\bOmega_i$ in Item 1 of \Cref{lem:classical_skew} we get $\Ex_I \left \| \P_{\bX_i \bY_i} - \P_{\bX_i \bY_i | W_C} \right \| \leq \sqrt{\delta}$, so we can replace the second instance of $\P_{\bX_i \bY_i}$ in~\eqref{eq:xi-change-y-1b} with $\P_{\bX_i \bY_i | W_C}$ (incurring an error of $\sqrt{\delta}$) to get
\[
\Ex_I \left \|  \P_{\bX_i \bY_i} \P_{\bOmega_i | W_C} \P_{\bR_\mi | \bOmega_i  W_C} - \P_{\bX_i \bY_i | W_C}  \P_{\bOmega_i | W_C} \P_{\bR_\mi | \bX_i \bY_i  W_C} \right \| \leq O(\sqrt{\delta}/\alpha^2)~.
\]
Marginalizing over $\bX_i \bY_i$ on both side yields
\begin{equation}\label{eq:claim24-1}
\Ex_I \left \| \P_{\bOmega_i | W_C} \P_{\bR_\mi | \bOmega_i  W_C} -  \P_{\bOmega_i | W_C} \P_{\bR_\mi | W_C} \right \| \leq O(\sqrt{\delta}/\alpha^2)~.
\end{equation}
We insert a copy of $\P_{\bY_i | \bOmega_i }$ in~\eqref{eq:claim24-1} to obtain via \Cref{lem:trivial}
\begin{equation}\label{eq:claim24-2}
\Ex_I \left \| \P_{\bOmega_i | W_C}\P_{\bR_\mi | \bOmega_i  W_C} \P_{\bY_i | \bOmega_i }  - \P_{\bOmega_i | W_C} \P_{\bR_\mi | W_C}  \P_{\bY_i | \bOmega_i } \right \| \leq O(\sqrt{\delta}/\alpha^2)~.
\end{equation}
Marginalizing over $\bX_i$ in Item 2 of~\Cref{lem:classical_skew} we obtain
\begin{equation}\label{eq:claim24-3}
\Ex_{I} \big\|\P_{\bR \bY_i | W_C} - \P_{\bR | W_C} \P_{\bY_i | \bOmega_i }\big\| \,\leq\,\sqrt{\delta}\;.
\end{equation}
Using $\bR = (\bR_\mi,\bOmega_i)$ it follows that $\P_{\bR | W_C}=\P_{\bOmega_i | W_C} \P_{\bR_\mi | \bOmega_i  W_C}$. Thus~\eqref{eq:claim24-3} implies that we can replace the first instance of $\P_{\bOmega_i | W_C} \P_{\bR_\mi | \bOmega_i  W_C}$ in~\eqref{eq:claim24-2} with $\P_{\bR \bY_i | W_C}$ (incurring an error of $\sqrt{\delta}$) to get
\begin{equation}\label{eq:claim24-4}
\Ex_I \left \| \P_{\bR \bY_i | W_C} -\P_{\bOmega_i | W_C} \P_{\bR_\mi | W_C}  \P_{\bY_i | \bOmega_i }  \right \| \leq O(\sqrt{\delta}/\alpha^2)~.
\end{equation}
Combining~\eqref{eq:claim24-4} with~\eqref{eq:claim24-1a} gives
\begin{equation}\label{eq:claim24-4c}
	\Ex_I \, \Ex_{\bOmega_i | W_C} \, \Ex_{\bR_\mi | W_C} \, \Ex_{\bY_i | \bOmega_i} \,\, \big \| \xi^{E_A}_{\br,\by_i} -  \xi^{E_A}_{\br} \big \|_1^2 \leq O(\sqrt{\delta}/\alpha^2)\;.
\end{equation}
Using the fact that $\Ex_I \| \P_{\bOmega_i | W_C} - \P_{\bOmega_i} \| \leq \sqrt{\delta}$, which follows from Item 1 of \Cref{lem:classical_skew}, recalling that $\bOmega_i = (\bD_i,\bM_i)$, conditioning on $\bD_i = A$ (which occurs with probability $1/2$) from~\eqref{eq:claim24-4c} we get
\begin{equation}\label{eq:claim24-4d}
	\Ex_I \, \Ex_{\bM_i | \bD_i =A} \, \Ex_{\bR_\mi | W_C} \, \Ex_{\bY_i | \bOmega_i=(A,\bM_i)} \,\, \big \| \xi^{E_A}_{\br,\by_i} -  \xi^{E_A}_{\br} \big \|_1^2 \leq O(\sqrt{\delta}/\alpha^2)\;.
\end{equation}
Next our goal is to exchange the expectations over $\bM_i | \bD_i = A$ and $\bY_i | \bOmega_i = (A,\bM_i)$ in~\eqref{eq:claim24-4d} with an expectation over $XY$. From the definition of the variables $(D,M,X,Y)$ (see \Cref{sec:p_setup}) we see that
\[\forall x\in \X\;,\qquad \P_X(x) \,\leq\, \frac{1 - \eta}{\alpha - \eta} \, \P_{M|D=A}(x)\;.\] 
Therefore for all $x \in \X$,
\begin{equation}
\label{eq:xi-change-y-1}
	\P_{XY}(x,y) = \P_{X}(x) \cdot \P_{Y|X=x}(y) \leq \frac{2}{\alpha} \,\P_{\bM_i | \bD_i = A}(x) \cdot \P_{\bY_i | \bOmega_i = (A,x)}(y)\;,
\end{equation}
where we used $\P_{\bY_i | \bOmega_i = (A,x)}(y) = \P_{Y|X=x}(y)$ and $\eta = \alpha/2$. Furthermore observe that when $\br = (\br_\mi,\bomega_i)$ with $\bomega_i = (A,x)$ then by definition
\begin{equation}
\label{eq:xi-change-y-2}
	\xi_{\br,y} = \xi_{\br_\mi,x,y} \qquad \text{and} \qquad \xi_{\br} = \xi_{\br_\mi,x}~.
\end{equation}
Thus
\begin{align}
	\Ex_I \, \Ex_{XY} \, \Ex_{\bR_\mi | W_C} \, \,\, \big \| \xi^{E_A}_{\br_\mi,x,y} -  \xi^{E_A}_{\br_\mi,x} \big \|_1^2 &\leq \frac{2}{\alpha} \, \Ex_I \, \Ex_{\bM_i \bY_i | \bD_i =A} \, \Ex_{\bR_\mi | W_C} \,\, \big \| \xi^{E_A}_{\br,\by_i} -  \xi^{E_A}_{\br} \big \|_1^2 \notag\\
	&\leq O(\sqrt{\delta}/\alpha^3)\;,\label{eq:xi-change-y-3}
\end{align}
where the first inequality is by~\eqref{eq:xi-change-y-1} and uses~\eqref{eq:xi-change-y-2} and the second is by~\eqref{eq:claim24-4d}. 
Conditioning on $Y = \dummy$, which occurs with probability at least $\alpha$ under the distribution $\mu_Y$, yields
\begin{equation}
\label{eq:xi-change-y-4}
	\Ex_I \, \Ex_{X} \, \Ex_{\bR_\mi | W_C} \, \,\, \big \| \xi^{E_A}_{\br_\mi,x,\dummy} -  \xi^{E_A}_{\br_\mi,x} \big \|_1^2 \leq O(\sqrt{\delta}/\alpha^4)~.
\end{equation}
Using the triangle inequality and $(a+b)^2 \leq 2(a^2+b^2)$ we get 
\begin{align*}
\Ex_I \, \Ex_{XY} \, \Ex_{\bR_\mi | W_C} \, \,\, \big \| \xi^{E_A}_{\br_\mi,x,\dummy} -  \xi^{E_A}_{\br_\mi,x,y} \big \|_1^2 &\leq \Ex_I \, \Ex_{XY} \, \Ex_{\bR_\mi | W_C} \, \,\, 2 \, \big \| \xi^{E_A}_{\br_\mi,x,\dummy} -  \xi^{E_A}_{\br_\mi,x} \big \|_1^2 + 2 \,\, \big \| \xi^{E_A}_{\br_\mi,x} -  \xi^{E_A}_{\br_\mi,x,y} \big \|_1^2 \\
& \leq  O(\sqrt{\delta}/\alpha^4)\;,
\end{align*}
where the second inequality uses~\eqref{eq:xi-change-y-3} and~\eqref{eq:xi-change-y-4}. This concludes the proof of~\eqref{eq:E_A_1} in \Cref{claim:xi-change-y}. The proof of~\eqref{eq:E_A_2} is similar.
\end{proof}

We are now ready to give the proof of Lemma~\ref{lem:unitary_bounds}.

\begin{proof}[Proof of Lemma~\ref{lem:unitary_bounds}]
We start by showing the existence of unitaries $V_{\br_\mi,y}$ that satisfy~\eqref{eq:vy_bound}. 
Let $\br_\mi = (\bomega_\mi,\ba_C,\bb_C)$ be such that it implies the event $W_C$, and let $x \in \X$, $y \in \Y$. 

\begin{claim}
\label{clm:xi-purification}
The state $\ket{\wt{\Phi}_\sspxy}^{E_A E_B}$ defined in~\eqref{eq:def-psitt} is a purification of the state $\xi_\ssxy^{E_A}$. Furthermore, the state $\ket{\wt{\Phi}_\sspy}^{E_A E_B}$ is a purification of the state $\xi_\sspy^{E_A}$. 
\end{claim}

\begin{proof}
Let $\bomega = (\bomega_\mi,\bomega_i)$ where $\bomega_i = (A,x)$. We can write $\xi_\ssxy$ explicitly as
\begin{align*}
	\xi_\ssxy = \frac{1}{\P_{\bOmega \bY_i \bA_C \bB_C}(\bomega,y,\ba_C,\bb_C)} \sum_{\by | \by_i = y} \P_{\bOmega \bY}(\bomega, \by) \, \puretomixed{\bomega, \by} \otimes \left (\sqrt{A_{\bomega}(\ba_C)} \otimes \sqrt{B_{\by}(\bb_C)} \right ) \left [\psi \right ] \otimes \puretomixed{\ba_C,\bb_C}\;,
\end{align*}
where $\br = (\br_\mi,\bomega_i)$. To see that the normalization is correct, the trace of the right-hand side (without the normalization) is
\begin{align*}
	\sum_{\by | \by_i = y} \P_{\bOmega \bY}(\bomega, \by) \, \bra{\psi} A_\bomega(\ba_C) \otimes B_\by(\bb_C) \ket{\psi} &= \sum_{\bx,\by | \by_i = y} \P_{\bOmega \bX \bY \bA_C \bB_C}(\bomega, \bx, \by, \ba_C, \bb_C) \\
	&= \P_{\bOmega \bY_i \bA_C \bB_C}(\bomega,y,\ba_C,\bb_C)~.
\end{align*}
Taking the partial trace of $\xi_\ssxy$ on the register $E_A$ we get
\begin{align}
\xi_\ssxy^{E_A} = &\frac{1}{\P_{\bOmega \bY_i \bA_C \bB_C}(\bomega,y,\ba_C,\bb_C)} \sum_{\by | \by_i = y} \P_{\bOmega \bY}(\bomega, \by) \, \sqrt{A_{\bomega}(\ba_C)} \,\, \sqrt{\rho} \,\, \overline{B}_\by(\bb_C) \,\, \sqrt{\rho} \,\, \sqrt{A_\bomega(\ba_C)}  \notag \\
&= \frac{\P_{\bOmega \bY_i}(\bomega, y)}{\P_{\bOmega \bY_i \bA_C \bB_C}(\bomega,y,\ba_C,\bb_C)} \, \sqrt{A_{\bomega}(\ba_C)} \,\, \sqrt{\rho} \,\, \Big ( \sum_{\by_\mi } \P_{\bY_\mi | \bomega_\mi}(\by_\mi) \overline{B}_\by(\bb_C) \Big) \,\, \sqrt{\rho} \,\, \sqrt{A_\bomega(\ba_C)} \notag \\
&= \frac{1}{\P_{\bA_C \bB_C | \bOmega = \bomega, \by_i = y}(\ba_C,\bb_C)} \, \sqrt{A_{\bomega}(\ba_C)} \,\, \sqrt{\rho} \,\,  \overline{B}_{\bomega_\mi, y} (\bb_C) \,\, \sqrt{\rho} \,\, \sqrt{A_\bomega(\ba_C)} \notag \\
&= \gamma_\sspxy^{-2} \, \sqrt{A_{\bomega}(\ba_C)} \,\, \sqrt{\rho} \,\,  \overline{B}_{\bomega_\mi, y} (\bb_C) \,\, \sqrt{\rho} \,\, \sqrt{A_\bomega(\ba_C)}
\;,\label{eq:xi-purification-1}
\end{align}
where the operator $B_{\bomega_\mi, y} (\bb_C)$ is defined in \Cref{sec:states-operators} and the last line follows from~\eqref{eq:gamma-def-3} in \Cref{prop:gamma}. Notice that since $\bomega_i = (A,x)$, the operator $A_\bomega(\ba_C)$ is equal to $A_{\bomega_\mi,\dummyx}(\ba_C)$ defined in~\eqref{eq:states-operators-3}. This implies that the state in~\eqref{eq:xi-purification-1} is also the reduced density matrix of $\ket{\wt{\Phi}_\sspxy}^{E_A E_B}$ (defined in~\eqref{eq:def-psitt}) on register $E_A$. 

The ``Furthermore'' part of the Claim follows from the fact that when $x = \dummy$, the operator $A_\bomega(\ba_C)$ for $\omega_i = (A,x)$ is equal to $A_{\bomega_\mi,\dummy}(\ba_C)$, defined in~\eqref{eq:states-operators-2}. Thus $\ket{\wt{\Phi}_\sspxy}^{E_A E_B} = \ket{\wt{\Phi}_\sspy}^{E_A E_B}$, which concludes the proof.
\end{proof}

\Cref{clm:xi-purification} implies that the states $\ket{\wt{\Phi}_\sspy}$ and $\ket{\wt{\Phi}_\sspp}$ purify $\xi^{E_A}_\sspy$ and $\xi^{E_A}_\sspp$ respectively. By Uhlmann's Theorem (\Cref{thm:uhlmann}) there exists a unitary $V_{\br_\mi, y}$ acting on $E_B$ such that 
\begin{align}
	\Ex_I \, \Ex_{\bR_\mi |  W_C} \, \, \Ex_{Y} \,\,  \bigbra{\wt{\Phi}_\sspy } V_{\br_\mi, y}  \bigket{\wt{\Phi}_\sspp}  
	&= \Ex_I \, \Ex_{\bR_\mi |  W_C} \, \, \Ex_{Y}  \,\, \F \big (\xi_\sspy^{E_A}, \xi_\sspp^{E_A} \big)  \notag \\
	& \geq 1 - \frac{1}{2} \Ex_I \, \Ex_{\bR_\mi |  W_C} \, \, \Ex_{Y}  \,\,  \big \| \xi^{E_A}_\sspy - \xi^{E_A}_\sspp \big \|_1 \notag \\
	&\geq 1 - \frac{1}{2} \sqrt {\Ex_I \, \Ex_{\bR_\mi |  W_C} \, \, \Ex_{Y}  \,\,  \big \| \xi^{E_A}_\sspy - \xi^{E_A}_\sspp \big \|_1^2} \notag \\
	& \geq 1 - O(\delta^{1/4}/\alpha^{5/2})\;,
\end{align}
where the second line follows from the Fuchs-van de Graaf inequality (Eq.~\eqref{eq:fuchs-graaf}) the third uses Jensen's inequality, and the fourth line uses Claim~\ref{claim:xi-change-y} by conditioning on $X = \dummy$, which occurs with probability $\alpha$. Translating from inner products to Euclidean distance and applying Jensen's inequality we get
\begin{align*}
	\Ex_I \, \Ex_{\bR_\mi |  W_C} \, \, \Ex_{Y} \,\,   \big \| \ket{\wt{\Phi}_\sspy } - V_{\br_\mi, y}  \ket{\wt{\Phi}_\sspp}  \big \|& \leq \sqrt{\Ex_I \, \Ex_{\bR_\mi |  W_C} \, \, \Ex_{Y} \,\,   \big \| \ket{\wt{\Phi}_\sspy } - V_{\br_\mi, y}  \ket{\wt{\Phi}_\sspp}  \big \|^2}\\ &= O\big(\delta^{1/8}/\alpha^{5/4}\big)~.
\end{align*}
Applying Markov's inequality over the index $i$ establishes~\eqref{eq:vy_bound}.

The argument for~\eqref{eq:vxy_bound} proceeds similarly. We start by using \Cref{clm:xi-purification}, which gives that the states $\ket{\wt{\Phi}_\sspxy}$ and $\ket{\wt{\Phi}_\sspxp}$ purify the reduced density matrices $\xi^{E_A}_\sspxy$ and $\xi^{E_A}_\sspxp$ respectively. Using Uhlmann's Theorem and Claim~\ref{claim:xi-change-y} in a similar way to how we derived~\eqref{eq:vy_bound} we deduce the existence of unitaries $V_{\ssxy}$ satisfying~\eqref{eq:vxy_bound}.

To prove~\eqref{eq:ux_bound} we use the following claim whose proof is analogous to that of \Cref{clm:xi-purification}.

\begin{claim}
\label{clm:lambda-purification}
The state $\ket{\wt{\Phi}_\ssxp}^{E_A E_B}$ defined in~\eqref{eq:def-psitt} is a purification of the state $\lambda_\ssxp^{E_B}$. 
\end{claim}

Using the claim we can use Uhlmann's Theorem again to deduce the existence of unitaries $U_{\br_\mi,x}$ that satisfy~\eqref{eq:ux_bound}. 
\end{proof}

\subsection{Second lemma}
\label{sec:l2}

 Recall that according to Proposition~\ref{prop:gamma} the normalization factors  $\gamma_\xy$ and $\gamma_{\pxy}$, when squared, correspond to the probabilities of obtaining outcomes $(\ba_C,\bb_C)$ conditioned on the dependency-breaking variable $\bomega_\mi$ and setting of the inputs $\bx_i,\by_i$ determined by the subscript of $\gamma$.

\begin{lemma}\label{lem:norms-are-close}
With probability at least $1 - O(\delta^{1/4})$ over the choice of $i \in [n] \setminus C$, 
	\begin{equation}
	\label{eq:norms-are-close-s0}
		\Ex_{XY} \, \Ex_{\bR_\mi | \dummy, \dummy, W_C} \, \Big | 1 - \frac{\gamma_{\br_\mi, x,y}}{\gamma_{\br_\mi,\dummy,\dummy}} \Big |^2 \leq O(\delta^{1/4}/\alpha^2)\;,
	\end{equation}
	and
	\begin{equation}
	\label{eq:norms-are-close-s1}
		\Ex_{XY} \, \Ex_{\bR_\mi | \dummy, \dummy, W_C} \, \Big | 1 - \frac{\gamma_\sspxy}{\gamma_{\br_\mi,\dummy,\dummy}} \Big |^2 \leq O(\delta^{1/4}/\alpha^3)\;,
	\end{equation}
	where the expectation over the random variable $\bR_\mi$ is conditioned on the events $\bX_i = \dummy$, $\bY_i = \dummy$, and $W_C$. 
\end{lemma}
Before proceeding with the proof we note that the denominator, $\gamma_\sspp$, cannot be zero in the expectation. This is because if $\br_\mi$ is sampled according to $\P_{\bR_\mi | \dummy, \dummy, W_C}$ with positive probability then $\P_{\bA_C \bB_C | \bomega_\mi, W_C}(\ba_C,\bb_C)$ must be nonzero. 

\begin{proof}
We first establish~\eqref{eq:norms-are-close-s0}. For this proof only we introduce the shorthand notation $\br_\mi = (\bomega_\mi,\ba_C,\bb_C) \in W_C$ to signify that the tuple $(\bx_C,\by_C,\ba_C,\bb_C)$ implies the event $W_C$, i.e.\ the tuple corresponds to winning questions and answers in the coordinates indexed by $C$). We start with the following claim.

\begin{claim}
With probability at least $1 - O(\delta^{1/4}/\alpha^2)$ over the choice of $i \in [n] \setminus C$,
\begin{equation}
\label{eq:norms-are-close-0}
\Ex_{XY} \sum_{\br_\mi \in W_C}  \Big | \P_{\bR_\mi | x,y}(\br_\mi) - \P_{\bR_\mi | \dummy, \dummy}(\br_\mi) \Big | = O\Big(\frac{\delta^{1/4}}{\alpha^2}\Big) \P(W_C | \bX_i = \dummy, \bY_i = \dummy)\;.
\end{equation}
\end{claim}

\begin{proof}
First note that
\begin{align}
\nonumber
\Ex_I \sum_{x,y} \P_{XY}(x,y) \left | \P (W_C| \bX_i = x,\bY_i = y)-\P(W_C) \right |
&= \P(W_C) \Ex_I  \left \| \P_{\bX_i \bY_i |W_C} - \P_{\bX_i\bY_i} \right \| \\
\label{eq:27-0}
&= O(\sqrt{\delta}) \cdot \P(W_C)\;,
\end{align}
where the second equality follows from the Item 1 of Lemma~\ref{lem:classical_skew}. Using $\P_{XY}(\bX_i = \dummy, \bY_i = \dummy) \geq \alpha^2$ we get that
\begin{equation}
\label{eq:w_c-doesnt-change}
\Ex_I \left | \P (W_C| \bX_i = \dummy,\bY_i = \dummy)-\P(W_C) \right | \leq O(\sqrt{\delta}/\alpha^2) \cdot \P(W_C)\;.
\end{equation}
Using the triangle inequality with~\eqref{eq:27-0} and~\eqref{eq:w_c-doesnt-change},
\begin{equation}\label{eq:27-1}
\Ex_I \, \Ex_{XY} \Big | \P(W_C| \bX_i = x,\bY_i = y)-\P(W_C | \bX_i = \dummy, \bY_i = \dummy) \Big | = O(\sqrt{\delta}/\alpha^2) \cdot \P(W_C)\;.
\end{equation}
Using the triangle inequality,  
\begin{align}
\Ex_I \, \Ex_{XY} \, \sum_{\br_\mi \in W_C} &\Big | \P(W_C) \cdot \P_{\bR_\mi | x, y, W_C}(\br_\mi) - \P_{\bR_\mi | x,y} (\br_\mi) \Big |\notag \\
 &\leq \Ex_I \, \Ex_{XY} \, \sum_{\br_\mi \in W_C} \Big | \big(\P(W_C|\bX_i = x,\bY_i = y) - \P(W_C)\big) \cdot \P_{\bR_\mi | x, y, W_C}(\br_\mi)  |\notag\\
&\qquad + \Ex_I \, \Ex_{XY} \, \sum_{\br_\mi \in W_C} \Big |  \P_{\bR_\mi \wedge W_C | x,y}(\br_\mi) - \P_{\bR_\mi | x,y}(\br_\mi)  \Big |  \notag \\
&= O(\sqrt{\delta}) \cdot \P(W_C)\;,\label{eq:27-1b}
\end{align}
where the last equality is obtained by using~\eqref{eq:27-0} to bound the first term and observing that the second term is $0$. 
Conditioning on $(X,Y) = (\dummy,\dummy)$ in~\eqref{eq:27-1b} we get
\begin{align}
\Ex_I \, \Ex_{XY} \, \sum_{\br_\mi\in W_C} \Big | \P(W_C) \cdot \P_{\bR_\mi |  \dummy,  \dummy, W_C}(\br_\mi) - \P_{\bR_\mi |\dummy,  \dummy}(\br_\mi)  \Big | = O\Big(\frac{\sqrt{\delta}}{\alpha^2}\Big) \P(W_C)~.\label{eq:27-1d}
\end{align}
Multiplying both sides of Item 4 of \Cref{lem:classical_skew} by $\P(W_C)$, using the second part of \Cref{lem:trivial}, and expanding the definition of $\| \P_{\bR_\mi | \dummy, \dummy, W_C} - \P_{\bR_\mi | x,y,W_C} \|$ we have
\begin{equation}
\label{eq:27-1c}
\Ex_I \, \Ex_{XY} \sum_{\br_\mi \in W_C} \left | \P(W_C) \cdot \P_{\bR_\mi | \dummy, \dummy, W_C}(\br_\mi) - \P(W_C) \cdot \P_{\bR_\mi | x,y, W_C}(\br_\mi) \right | \leq O\Big(\frac{\sqrt{\delta}}{\alpha^2}\Big) \P(W_C)
\end{equation}
Using the triangle inequality we thus obtain
\begin{align*}
&\Ex_I \Ex_{XY} \sum_{\br_\mi \in W_C}  \Big | \P_{\bR_\mi | x,y}(\br_\mi) - \P_{\bR_\mi | \dummy, \dummy}(\br_\mi) \Big | \\
&\leq \Ex_I \Ex_{XY} \sum_{\br_\mi \in W_C}  \Big | \P_{\bR_\mi | x,y}(\br_\mi) - \P(W_C) \cdot \P_{\bR_\mi | x,y, W_C}(\br_\mi) \Big | \\
& \qquad \qquad \qquad +  \Big | \P(W_C) \cdot \P_{\bR_\mi | x,y, W_C}(\br_\mi) - \P(W_C) \cdot \P_{\bR_\mi | \dummy,\dummy, W_C}(\br_\mi) \Big | \\
& \qquad \qquad \qquad \qquad + \Big | \P(W_C) \cdot \P_{\bR_\mi | \dummy,\dummy, W_C}(\br_\mi) - \P_{\bR_\mi | \dummy,\dummy}(\br_\mi) \Big |  \\
&= O\Big(\frac{\sqrt{\delta}}{\alpha^2}\Big) \P(W_C)
\end{align*}
where we bound the first term of the second line using~\eqref{eq:27-1b}, the second term of the second line using~\eqref{eq:27-1c}, and the third term using~\eqref{eq:27-1d}~.
Using Markov's inequality, with probability at least $1 - \kappa$ over $i \in [n] \setminus C$, we have
\[
\Ex_{XY} \sum_{\br_\mi \in W_C}  \Big | \P_{\bR_\mi | x,y}(\br_\mi) - \P_{\bR_\mi | \dummy, \dummy}(\br_\mi) \Big | = O\Big(\frac{\sqrt{\delta}}{\kappa \cdot \alpha^2}\Big) \P(W_C)\;.
\]
All that remains is to show that with high probability over the index $i$, $\P(W_C)$ is not too far from $\P(W_C | \bX_i = \dummy, \bY_i = \dummy)$. From~\eqref{eq:w_c-doesnt-change} and Markov's inequality we have that with probability at least $1 - \kappa$ over $i \in [n] \setminus C$, we have 
\[
	\Big (1 - O\Big(\frac{\sqrt{\delta}}{\kappa \cdot \alpha^2}\Big) \Big)  \cdot \P(W_C) \leq \P (W_C| \bX_i = \dummy,\bY_i = \dummy) \leq \Big (1 + O\Big(\frac{\sqrt{\delta}}{\kappa \cdot \alpha^2}\Big) \Big) \cdot \P(W_C)~.
\]
Setting $\kappa = \delta^{1/4}/2$, we obtain~\eqref{eq:norms-are-close-0}
\end{proof}

We proceed with the proof of Lemma~\ref{lem:norms-are-close}. 
Observe that we can write
	\begin{align*}
		\P_{\bR_\mi | \dummy, \dummy, W_C}(\br_\mi) &= \P(W_C |  \bX_i = \dummy, \bY_i = \dummy)^{-1} \cdot \P_{\bOmega_\mi \bA_C \bB_C \wedge W_C | \dummy, \dummy}(\bomega_\mi,\ba_C,\bb_C) \\
		&= \P(W_C | \bX_i = \dummy, \bY_i = \dummy)^{-1} \cdot \P_{\bOmega_\mi \bA_C \bB_C | \dummy, \dummy}(\bomega_\mi,\ba_C,\bb_C) \\
		&= \P(W_C | \bX_i = \dummy, \bY_i = \dummy)^{-1} \cdot \P_{\bOmega_\mi | \dummy, \dummy}(\bomega_\mi) \cdot \P_{\bA_C \bB_C | \bomega_\mi,\dummy, \dummy}(\ba_C,\bb_C) \\
		&= \P(W_C |\bX_i = \dummy, \bY_i = \dummy)^{-1} \cdot \P_{\bOmega_\mi}(\bomega_\mi) \cdot \gamma_{\br_\mi,\dummy,\dummy}^2 \;,
	\end{align*}
	where the second line follows from the fact that, given $\br_\mi \in W_C$, the tuple $(\bx_C,\by_C,\ba_C,\bb_C)$ automatically implies the event $W_C$ and the last line follows from the fact that $\P_{\bOmega_\mi | \bX_i,\bY_i} = \P_{\bOmega_\mi}$, because coordinates of $\bOmega$ are independent (Item 1 in Claim~\ref{claim:ufacts}).
	
	Fix $i \in [n] \setminus C$. Using that $|a - b|^2 \leq |a^2 - b^2|$ for all $a,b \geq 0$, 
	\begin{align}
		&\Ex_{XY} \, \Ex_{\bR_\mi | \dummy, \dummy, W_C} \, \Big | 1 - \frac{\gamma_{\br_\mi, x,y}}{\gamma_{\br_\mi,\dummy,\dummy}} \Big |^2 \notag \\
		&\leq \Ex_{XY} \, \Ex_{\bR_\mi | \dummy, \dummy, W_C} \, \Big | 1 - \frac{\gamma_{\br_\mi, x,y}^2}{\gamma_{\br_\mi,\dummy,\dummy}^2} \Big | \label{eq:nac-1} \\
		&= \P(W_C |\bX_i = \dummy, \bY_i = \dummy)^{-1} \,\Ex_{XY} \,\sum_{\br_\mi \in W_C} \P_{\bOmega_\mi}(\bomega_\mi) \cdot \gamma_{\br_\mi,\dummy,\dummy}^2 \cdot \, \Big | 1 - \frac{\gamma_{\br_\mi, x,y}^2}{\gamma_{\br_\mi,\dummy,\dummy}^2} \Big | \notag \\
		&= \P(W_C |\bX_i = \dummy, \bY_i = \dummy)^{-1} \, \Ex_{XY} \,\sum_{\br_\mi \in W_C} \P_{\bOmega_\mi}(\bomega_\mi) \cdot \, \Big | \gamma_{\br_\mi,\dummy,\dummy}^2 - \gamma_{\br_\mi, x,y}^2 \Big | \notag \\
		&= \P(W_C |\bX_i = \dummy, \bY_i = \dummy)^{-1} \, \Ex_{XY} \,\sum_{\br_\mi \in W_C} \, \Big | \P_{\bR_\mi | \dummy,\dummy}(\br_\mi) - \P_{\bR_\mi | x,y}(\br_\mi) \Big |~. \label{eq:norms-are-close-0a}
	\end{align}
	In the second-to-last line we used the fact that $\P_{\bOmega_\mi}(\bomega_\mi) = \P_{\bOmega_\mi | \bX_i = x, \bY_i = y}(\bomega_\mi)$ for all $x,y$, which again is because coordinates of $\Omega$ are independent.  

Combining~\eqref{eq:norms-are-close-0} with~\eqref{eq:norms-are-close-0a} shows~\eqref{eq:norms-are-close-s0}. To establish~\eqref{eq:norms-are-close-s1} we notice that for all $i, \br_\mi$, we have
\begin{align*}
	\gamma_\sspxy^2 &= \bra{\psi} A_{\bomega_\mi,\dummy/x}(\ba_C) \otimes B_{\bomega_\mi,y}(\bb_C) \ket{\psi} \\
	&= \eta \, \bra{\psi} A_{\bomega_\mi,\dummy}(\ba_C) \otimes B_{\bomega_\mi,y}(\bb_C) \ket{\psi} + (1 - \eta) \, \bra{\psi} A_{\bomega_\mi,x}(\ba_C) \otimes B_{\bomega_\mi,y}(\bb_C) \ket{\psi} \\
	&= \eta \, \gamma_\sspy^2 + (1 - \eta) \,\gamma_\ssxy^2\;,
\end{align*}
where the second line uses the definition~\eqref{eq:states-operators-3}. 
Therefore
\begin{align*}
\Ex_{XY} \, \Ex_{\bR_\mi | \dummy, \dummy, W_C} \, \Big | 1 - \frac{\gamma_{\br_\mi, \dummy/x,y}}{\gamma_{\br_\mi,\dummy,\dummy}} \Big |^2 &\leq \Ex_{XY} \, \Ex_{\bR_\mi | \dummy, \dummy, W_C} \, \Big | 1 - \frac{\gamma_{\br_\mi, \dummy/x,y}^2}{\gamma_{\br_\mi,\dummy,\dummy}^2} \Big | \\
&\leq \Ex_{XY} \, \Ex_{\bR_\mi | \dummy, \dummy, W_C} \, \eta \,\Big | 1 - \frac{\gamma_{\br_\mi, \dummy ,y}^2}{\gamma_{\br_\mi,\dummy,\dummy}^2} \Big |+(1 - \eta)\,\Big | 1 - \frac{\gamma_{\br_\mi, x,y}^2}{\gamma_{\br_\mi,\dummy,\dummy}^2} \Big |\;,
\end{align*}
where the second line is by the triangle inequality. 
Using $\P_X(X = \dummy) \geq \alpha$ the two terms in this last expression can be bounded by $O(\delta^{1/4}/\alpha^3)$ using the proof of~\eqref{eq:norms-are-close-s0}, which in fact bounds the stronger quantity~\eqref{eq:nac-1}. This concludes the proof of \Cref{lem:norms-are-close}. 

\end{proof}

\subsection{Proof of Proposition~\ref{prop:local_unitaries}}

We end the section with the proof of Proposition~\ref{prop:local_unitaries}. 

\begin{proof}[Proof of Proposition~\ref{prop:local_unitaries}]
For every $i,\br_\mi$, $x$ and $y$ let unitaries $U_{\br_\mi, x}$, $V_{\br_\mi, y}$ and $V_{\ssxy}$ be as in Lemma~\ref{lem:unitary_bounds}. For notational convenience we suppress the dependence on $(i, \br_\mi)$; thus the unitaries $U_x, V_y, V_{x,y}$, the states $\ket{\Phi_{x,y}}$, and their normalizations $\gamma_{x,y}$ all implicitly depend on $i$ and $\br_\mi$. 

We call an index $i \in [n] \setminus C$ \emph{good} if it satisfies (i) the conclusions of \Cref{lem:unitary_bounds}, (ii) the conclusions of \Cref{lem:norms-are-close}, and (iii) it holds that
\begin{equation}\label{eq:good-i-cond}
\left \| \P_{\bR_\mi | \bX_i = \dummy, \bY_i = \dummy,W_C} - \P_{\bR_\mi | W_C} \right \| \,\leq\, O(\delta^{1/4}/\alpha^2)\;.
\end{equation}
Applying the data processing inequality (\Cref{lem:data-processing}) to Item 3 of \Cref{lem:classical_skew} to marginalize over the random variable $\bOmega_i$, and then applying Markov's inequality over the index $i$, we get that~\eqref{eq:good-i-cond} holds with probability at least $1 - \delta^{1/4}$ over a uniformly random choice of $i$. This combined with \Cref{lem:unitary_bounds} and \Cref{lem:norms-are-close} implies that an index $i$ is good with probability at least $1 - O(\delta^{1/16})$. 
Using the definition we have that 
\begin{align}
\Ex_{\bR_\mi | \dummy, \dummy, W_C} \, \Ex_{XY} \,  \big \|\ket{\wt{\Phi}_{x,y}} - \gamma_{\dummy,\dummy}^{-1}\ket{\Phi_{x,y}} \big \| 
&= \Ex_{\bR_\mi | \dummy, \dummy, W_C} \,\Ex_{XY} \, \Big|1 - \frac{\gamma_{x,y}}{\gamma_{\dummy,\dummy}} \Big| \notag\\
& = O(\delta^{1/8}/\alpha)~,\label{eq:local_unitaries-0}
\end{align}
where the second line is by Jensen's inequality and~\eqref{eq:norms-are-close-s0} in \Cref{lem:norms-are-close}. Similarly,
\begin{align}
\Ex_{\bR_\mi | \dummy, \dummy, W_C} \, \Ex_{XY} \,  \big \|\ket{\wt{\Phi}_{\pxy}} - \gamma_{\dummy,\dummy}^{-1}\ket{\Phi_{\pxy}} \big \| &= \Ex_{\bR_\mi | \dummy, \dummy, W_C} \,\Ex_{XY} \, \Big|1 - \frac{\gamma_{\pxy}}{\gamma_{\dummy,\dummy}} \Big|\notag\\
& = O\big(\delta^{1/8}/\alpha^{3/2}\big)~,\label{eq:local_unitaries-0b}
\end{align}
by~\eqref{eq:norms-are-close-s1}. We note that in the above division by $\gamma_\pp$ is well-defined because $\br_\mi$ is sampled with positive probability from the distribution $\P_{\bR_\mi | \dummy,\dummy,W_C}$. Using~\eqref{eq:good-i-cond} we get from the bounds in \Cref{lem:unitary_bounds} that 
\begin{align}
	\Ex_{\bR_\mi | \dummy,\dummy,W_C} \,\, \Ex_X \,\, \,\,\, \big \| \ket{\wt{\Phi}_\xp } - (U_{x} \otimes \Id)  \ket{\wt{\Phi}_\pp}  \big \| &=  O(\delta^{1/16}/\alpha^2)\;,\label{eq:ux_bound-2}\\
	\Ex_{\bR_\mi | \dummy,\dummy,W_C} \,\, \Ex_Y \,\, \,\,\, \big \| (\Id \otimes V_{y})  \ket{\wt{\Phi}_\pp}  - \ket{\wt{\Phi}_\py }  \big \| &=  O(\delta^{1/16}/\alpha^2)\;,\label{eq:vy_bound-2}\\
	\Ex_{\bR_\mi | \dummy,\dummy,W_C} \,\, \Ex_{XY} \,\,  \big \|  (\Id \otimes V_{\xy})  \ket{\wt{\Phi}_\pxy}  - \ket{\wt{\Phi}_\pxp }  \big \| &=  O(\delta^{1/16}/\alpha^2) \label{eq:vxy_bound-2}
		\end{align}
	where we bound $O(\delta^{1/16}/\alpha^{5/4}) + O(\delta^{1/14}/\alpha^2) = O(\delta^{1/16}/\alpha^2)$. 

The main step in the proof is provided by the following claim. 

\begin{claim}\label{claim:good-i-2}
It holds that
\begin{align}
 \Ex_{\bR_\mi | \dummy, \dummy, W_C} \, \Ex_{XY} \, \big \| (U_{x} \otimes V_{y}) \ket{\wt{\Phi}_{\pp}} - \ket{\wt{\Phi}_{\xy}} \big\|  
 &\leq \Ex_{\bR_\mi | \dummy, \dummy, W_C} \, \Ex_{XY} \, \gamma_\pp^{-1} \, \left \| V_y \ket{{\Phi}_{\dummy,\dummy}} - \ket{{\Phi}_{\dummy,y}} \right \|\label{eq:good-i-2a}\\
 &\qquad + 2\eta^{-1/2} \gamma_\pp^{-1} \,\big \|  V_{x,y} \ket{\Phi_{\dummyx,y}} - \ket{{\Phi}_{\dummyx,\dummy}} \big \| \label{eq:good-i-2b}\\
&  \qquad + \gamma_\pp^{-1} \, \left \|  U_x \ket{{\Phi}_{\dummy,\dummy}} -  \ket{{\Phi}_{x,\dummy}} \right \| + O(\delta^{1/8}/\alpha)\;.\label{eq:good-i-2c}
\end{align}
\end{claim}

\begin{proof}
We start by writing 
\begin{align}
\Ex_{\bR_\mi | \dummy, \dummy, W_C} \, \Ex_{XY} \,& \big \| (U_{x} \otimes V_{y}) \ket{\wt{\Phi}_{\pp}} - \ket{\wt{\Phi}_{\xy}} \big\| \notag \\
&\leq  \Ex_{\bR_\mi | \dummy, \dummy, W_C} \, \Ex_{XY} \, \left \| (U_{x} \otimes V_{y}) \bigket{\wt{\Phi}_{\pp}} - \frac{\gamma_\xy}{\gamma_\pp} \bigket{\wt{\Phi}_{\xy}} \right\| +  \left \| \frac{\gamma_\xy}{\gamma_\pp} \bigket{\wt{\Phi}_{\xy}} - \bigket{\wt{\Phi}_{\xy}} \right\| \notag \\
&= \Ex_{\bR_\mi | \dummy, \dummy, W_C} \, \Ex_{XY} \, \gamma_\pp^{-1} \left \| (U_{x} \otimes V_{y}) \bigket{{\Phi}_{\pp}} - \bigket{{\Phi}_{\xy}} \right\| +  \left | \frac{\gamma_\xy}{\gamma_\pp} -1  \right| \notag \\
&\leq \Ex_{\bR_\mi | \dummy, \dummy, W_C} \, \Ex_{XY} \, \gamma_\pp^{-1} \left \| (U_{x} \otimes V_{y}) \bigket{{\Phi}_{\pp}} - \bigket{{\Phi}_{\xy}} \right\| + O(\delta^{1/8}/\alpha) \;,\label{eq:local_unitaries-1}
\end{align}
where the last line follows from~\eqref{eq:local_unitaries-0}.

For $s \in \{\dummy, \dummyx, x\}$ and $y \in \Y$ recall that the unnormalized state $\ket{\Phi_{s,y}}$ is defined in~\eqref{eq:def-state-y} as $A_s^{1/2} \otimes B_y^{1/2} \ket{\psi}$. Using~\eqref{eq:states-operators-3a} (resp.~\eqref{eq:states-operators-3b}) it follows that $A_x A_{\dummyx}^{-1/2} A_\dummyx^{1/2} = A_x$ (resp. $A_\dummy A_{\dummyx}^{-1/2} A_\dummyx^{1/2} = A_\dummy$), because the image of $A_x$ (resp. $A_\dummy$) is contained in the image of $A_\dummyx$. Thus
\begin{align}
	&\left \| U_x \ket{\Phi_{\dummy,y}} - \ket{\Phi_{x,y}} \right \| \notag \\
	&= \big \| U_x A_\dummy^{1/2} A_\dummyx^{-1/2} \ket{\Phi_{\dummyx,y}} - A_{x}^{1/2} A_{\dummy\! /x}^{-1/2}  \ket{\Phi_{\dummyx,y}} \big \| \notag \\
	&= \big \| U_x A_\dummy^{1/2} A_\dummyx^{-1/2} \otimes V_{x,y} \ket{\Phi_{\dummyx,y}} - A_{x}^{1/2} A_{\dummy\! /x}^{-1/2} \otimes V_{x,y}  \ket{\Phi_{\dummyx,y}} \big \| \notag \\
&\leq  \big \| \big ( U_x A_\dummy^{1/2} A_\dummyx^{-1/2} \big ) \otimes V_{x,y} \ket{\Phi_{\dummyx,y}} - \big ( U_x A_\dummy^{1/2} A_\dummyx^{-1/2} \big ) \ket{\Phi_{\dummyx, \dummy}} \big \|  \label{eq:22-2a}\\
& \qquad +  \big \| \big ( U_x A_\dummy^{1/2} A_\dummyx^{-1/2} \big ) \ket{\Phi_{\dummyx,\dummy}} - A_x^{1/2} A_\dummyx^{-1/2}  \ket{\Phi_{\dummyx,\dummy}} \big \| \label{eq:22-2b}\\
& \qquad + \big \| A_x^{1/2} A_\dummyx^{-1/2} \ket{\Phi_{\dummyx,\dummy}} - A_x^{1/2} A_\dummyx^{-1/2}  \otimes V_{x,y} \ket{\Phi_{\dummyx,y}} \big \|~,\label{eq:22-2c}
\end{align}
by the triangle inequality. We bound each of these three terms as follows. 
Using $\|A_\dummy^{1/2} A_\dummyx^{-1/2}\|\leq \eta^{-1/2}$, which follows from~\eqref{eq:states-operators-3a} and operator motonicity of the square root, the term~\eqref{eq:22-2a} can be bounded as
$$
\left \| \left ( U_x A_\dummy^{1/2} A_\dummyx^{-1/2} \right ) \otimes V_{x,y}\, \ket{\Phi_{\dummyx,y}} - \left ( U_x A_\dummy^{1/2} A_\dummyx^{-1/2} \right ) \ket{\Phi_{\dummyx,\dummy}} \right \|  \leq \eta^{-1/2} \, \left \|  V_{x,y} \ket{\Phi_{\dummyx,y}} - \ket{\Phi_{\dummyx,\dummy}} \right \|.
$$
The term~\eqref{eq:22-2b} can be re-written as
$$
\left \| \left ( U_x A_\dummy^{1/2} A_\dummyx^{-1/2} \right ) \ket{\Phi_{\dummyx,\dummy}} - A_x^{1/2} A_\dummyx^{-1/2}  \ket{\Phi_{\dummyx,\dummy}} \right \| = \left \| U_x \ket{\Phi_{\dummy,\dummy}} - \ket{\Phi_{x,\dummy}} \right \|.
$$
Finally, using $\|A_x^{1/2} A_\dummyx^{-1/2}\|\leq (1 - \eta)^{-1/2}$ from~\eqref{eq:states-operators-3b} and that $(1 - \eta)^{-1/2} \leq \eta^{-1/2}$ because $\eta \leq 1/2$, the term~\eqref{eq:22-2c} can be bounded as
$$
\left \| A_x^{1/2} A_\dummyx^{-1/2} \ket{\Phi_{\dummyx,\dummy}} - A_x^{1/2} A_\dummyx^{-1/2}  \otimes V_{x,y}\, \ket{\Phi_{\dummyx,y}} \right \| \leq \eta^{-1/2} \left \|  \ket{\Phi_{\dummyx,\dummy}} - V_{x,y}\, \ket{\Phi_{\dummyx,y}} \right \|.
$$
Putting the three bounds together, from~~\eqref{eq:22-2a}--\eqref{eq:22-2c} we get
\begin{align}
\left \| U_x \ket{\Phi_{\dummy,y}} - \ket{\Phi_{x,y}} \right \| \leq   2 \eta^{-1/2} \big \|  V_{x,y} \ket{\Phi}_{\dummyx,y} - \ket{{\Phi}_{\dummyx,\dummy}} \big \| + \left \|  U_x \ket{{\Phi}_{\dummy,\dummy}} -  \ket{{\Phi}_{x,\dummy}} \right \| .\label{eq:22-3}
\end{align}
Using the triangle inequality and that $U_x$ is unitary,
\begin{align*}
&\left \| (U_x \otimes V_y) \ket{{\Phi}_{\dummy,\dummy}} - \ket{{\Phi}_{x,y}} \right\| \\
&\leq  \left \| V_y \ket{{\Phi}_{\dummy,\dummy}} - \ket{{\Phi}_{\dummy,y}} \right \| + \left \| U_x \ket{{\Phi}_{\dummy,y}} - \ket{{\Phi}_{x,y}} \right \| \notag\\ 
&\leq  \left \| V_y \ket{{\Phi}_{\dummy,\dummy}} - \ket{{\Phi}_{\dummy,y}} \right \|  + 2\eta^{-1/2} \left \|  V_{x,y} \ket{\Phi_{\dummyx,y}} - \ket{{\Phi}_{\dummyx,\dummy}} \right \| +  \left \|  U_x \ket{{\Phi}_{\dummy,\dummy}} -  \ket{{\Phi}_{x,\dummy}} \right \|\;,
\end{align*}
where the last inequality is due to~\eqref{eq:22-3}. Inserting into~\eqref{eq:local_unitaries-1} proves the claim. 
\end{proof}

To conclude the proof \Cref{prop:local_unitaries} it remains to bound each of the three terms on the right-hand side of Claim~\ref{claim:good-i-2} by $O(\delta^{1/16}/\alpha^3)$, and then use~\eqref{eq:good-i-cond} to exchange the expectation $ \Ex_{\bR_\mi | \dummy, \dummy, W_C}$ with $ \Ex_{\bR_\mi | W_C}$ by introducing an additive $O(\delta^{1/4}/\alpha^2)$ error. 
We start with bounding~\eqref{eq:good-i-2a}:
\begin{align*}
\Ex_{\bR_\mi | \dummy, \dummy, W_C} \, \Ex_Y \, \gamma_\pp^{-1} &\,\, \left \|  V_y \ket{{\Phi}_{\dummy,\dummy}} -  \ket{{\Phi}_\py} \right \|  \\
&= \Ex_{\bR_\mi | \dummy, \dummy, W_C} \, \Ex_Y \,\, \Big \|  V_y \ket{\wt{\Phi}_{\dummy,\dummy}} -  \frac{\gamma_\py}{\gamma_\pp} \ket{\wt{\Phi}_\py} \Big \| \\
	&\leq  \Ex_{\bR_\mi | \dummy, \dummy, W_C} \, \Ex_Y \, \big \|  V_y \ket{\wt{\Phi}_{\dummy,\dummy}} - \ket{\wt{\Phi}_{\py}} \big \| + \Big \| \ket{\wt{\Phi}_{\py}} - \frac{\gamma_{\py}}{\gamma_\pp} \ket{\wt{\Phi}_{\py}}\Big \|\\
&= O(\delta^{1/16}/\alpha^2) + \Ex_{\bR_\mi | \dummy, \dummy, W_C} \, \Ex_Y \, \Big | 1 - \frac{\gamma_{\py}}{\gamma_\pp} \Big | \\
&= O(\delta^{1/16}/\alpha^2) + O(\delta^{1/8}/\alpha^2) = O(\delta^{1/16}/\alpha^2)~,
\end{align*}
	where the third line uses~\eqref{eq:ux_bound-2} to bound the first term and the last line follows from~\eqref{eq:local_unitaries-0} and conditioning on $X = \dummy$, which occurs with probability $\alpha$. We bound~\eqref{eq:good-i-2c} in an analogous fashion. 
	Finally, we bound~\eqref{eq:good-i-2b} as follows:
\begin{align*}
&2\eta^{-1/2}  \, \Ex_{\bR_\mi | \dummy, \dummy, W_C} \, \Ex_{XY} \, \, \gamma_\pp^{-1} \,\, \left \|  V_{x,y} \ket{{\Phi}_{\pxy}} -  \ket{{\Phi}_\pxp} \right \| \\
&= 2\eta^{-1/2} \,  \Ex_{\bR_\mi | \dummy, \dummy, W_C} \, \Ex_{XY} \,\, \left \|  \frac{\gamma_\pxy}{\gamma_\pp} V_{x,y} \ket{\wt{\Phi}_{\pxy}} -  \frac{\gamma_\pxp}{\gamma_\pp} \ket{\wt{\Phi}_\pxp} \right \| \\
	&\leq  2\eta^{-1/2} \,\Ex_{\bR_\mi | \dummy, \dummy, W_C} \, \Ex_{XY} \, \Big \| \frac{\gamma_\pxy}{\gamma_\pp} \ket{\wt{\Phi}_{\pxy}} - \ket{\wt{\Phi}_{\pxy}} \Big \| + \Big \|  V_{x,y} \ket{\wt{\Phi}_{\pxy}} - \ket{\wt{\Phi}_{\pxp}} \Big \|  \\
	& \qquad \qquad \qquad \qquad \qquad \qquad \qquad \qquad + \left \| \ket{\wt{\Phi}_{\pxp}} - \frac{\gamma_{\pxp}}{\gamma_\pp} \ket{\wt{\Phi}_{\pxp}}\right \|\\
&= 2\eta^{-1/2} \, \Ex_{\bR_\mi | \dummy, \dummy, W_C} \, \Ex_{XY} \, \Big | 1 - \frac{\gamma_{\pxy}}{\gamma_\pp} \Big | + O(\delta^{1/16}/(\eta^{1/2} \cdot \alpha^2)) + 2\eta^{-1/2} \, \Ex_{\bR_\mi | \dummy, \dummy, W_C} \, \Ex_{XY} \, \Big | 1 - \frac{\gamma_{\pxp}}{\gamma_\pp} \Big | \\
&= O(\delta^{1/8}/\alpha^2) + O(\delta^{1/16}/\alpha^2) + O(\delta^{1/8}/\alpha^3) = O(\delta^{1/16}/\alpha^3)~.
\end{align*}
The last line follows from $\eta = \alpha/2$,~\eqref{eq:local_unitaries-0b} to bound the first term,~\eqref{eq:vxy_bound-2} to bound the second term, and ~\eqref{eq:local_unitaries-0b} along with conditioning on $Y = \dummy$ to bound the last term.

\end{proof}

\section{Parallel repetition of anchored games}
\label{sec:analysis}

In this section we show our main theorem, which we state as follow.

\begin{theorem}[Main Theorem]\label{thm:anchorpr_quantum}
	There exists a universal constant $0 < c < \frac{1}{16 \, \log (e)}$ such that the following holds. Let $0<\alpha\leq 1$ and let $G$ be an $\alpha$-anchored game. Then for all $0 < \eps \leq 1$,
for all integers $n$, and for all $p$ satisfying
\begin{equation}
\label{eq:p}
p \geq \frac{4}{\eps} \, \exp \Big ( - \frac{c \cdot \alpha^{48} \cdot \eps^{17} \cdot n}{s} \Big)\;,
\end{equation}
where  $s = \max \{ \log |\A \times \B|, 1 \}$, 
\[
	\E(G^n,p) \geq \E(G,1 - \eps)\;.
\]
\end{theorem}

We show that Theorem~\ref{thm:main-informal} stated in the introduction follows as an easy corollary. 

\begin{proof}[Proof of Theorem~\ref{thm:main-informal}]
	Suppose that $\eval(G^n)$ is larger than 
	\[p = \frac{4}{\eps} \, \exp \Big ( - \frac{c \cdot \alpha^{48} \cdot \eps^{17} \cdot n}{s} \Big)\;.\]
	Then $\E(G^n,p) < \infty$, which by Theorem~\ref{thm:anchorpr_quantum} and the choice of $p$ implies that $\E(G,1 - \eps) < \infty$. This contradicts the assumption that $\eval(G) < 1 - \eps$.
\end{proof}

\subsection{Main lemma}

As in Section~\ref{sec:quantum_setup} and Section~\ref{sec:unitaries} fix an $\alpha$-anchored two-player game $G = (\X \times \Y,\A \times \B,\mu,V)$ and a strategy $\strategy^n = (\ket{\psi},A,B)$ for $G^n$. 
For every choice of subset $C \subseteq [n]$, index  $i \in [n] \setminus C$, and $\br_\mi = (\bomega_\mi,\ba_C,\bb_C)$ we define a  strategy $\strategy_{\br_\mi}$ for $G$ as follows. The shared state is 
\[\ket{\wt{\Phi}_\sspp} \,\in\, \C_{E_A}^d \otimes \C_{E_B}^d\;,\]
as defined in~\eqref{eq:def-psitt}. The measurement operators used by the players are
\[ \wt{A}_x(a)\,=\,  U_{\br_\mi,x}^\dagger \what{A}_{\br_\mi,x}(\ba_i)U_{\br_\mi,x}\qquad\text{and}\qquad \wt{B}_y(b) \,=\,V_{\br_\mi,y}^\dagger\what{B}_{\br_\mi, y}(\bb_i) V_{\br_\mi,y}\;,\]
where $x,y$ denote questions to the first and second player respectively and the answers $a\in \A$, $b\in \B$ are identified with $\ba_i$ and $\bb_i$ on the right-hand side (as will generally be the case in this and the following sections). 
Here, $U_{\br_\mi,x}$, $V_{\br_\mi,y}$ are the unitaries from Proposition~\ref{prop:local_unitaries}
 and $\{\what{A}_{\br_\mi,x}(\ba_i) \}$ and $\{\what{B}_{\br_\mi, y}(\bb_i)\}$ POVMs defined as
\begin{align*}
	\what{A}_{\br_\mi,x}(\ba_i) = \sum_{\ba | \ba_i, \ba_C} (A_{\bomega_\mi,x}(\ba_C) )^{-1/2} \cdot A_{\bomega_\mi,x}(\ba) \cdot (A_{\bomega_\mi,x}(\ba_C) )^{-1/2} \;,\\
	\what{B}_{\br_\mi, y}(\bb_i) = \sum_{\bb | \bb_i, \bb_C}  (B_{\bomega_\mi,y}(\bb_C) )^{-1/2} \cdot B_{\bomega_\mi,y}(\bb) \cdot (B_{\bomega_\mi,y}(\bb_C) )^{-1/2}\;,
\end{align*}
where $\ba | \ba_i, \ba_C$ (resp. $\bb | \bb_i, \bb_C$) denotes summing over tuples $\ba$ that are consistent with $\ba_C$ and $\ba_i$ (resp. $\bb$ that are consistent with $\bb_C$ and $\bb_i$). 
Let $\Qsf_{A B | \br_\mi, x,y}$ denote the distribution of answers $(a,b)$ using the strategy $\strategy_{\br_\mi}$ when the players are given question pair $(x,y)$. In other words,
\begin{align} 
\Qsf_{A B | \br_\mi, x,y}(a,b)&= \Tr \Big ( \wt{A}_x(a) \otimes \wt{B}_y(b) \,\, \wt{\Phi}_\sspp \Big) \, \notag \\
&= \Tr \Big ( \what{A}_{\br_\mi,x} (a) \otimes \what{B}_{\br_\mi,y} (b) \,\, (U_{\br_\mi,x} \otimes V_{\br_\mi,y}) \, \wt{\Phi}_\sspp \, (U_{\br_\mi,x} \otimes V_{\br_\mi,y})^\dagger \Big) \; .\label{eq:main-lemma-0}
\end{align}
The following provides the main step in the proof of Theorem~\ref{thm:anchorpr_quantum}.

\begin{lemma}[Main Lemma]
\label{lem:anchorpr_main_lemma}
There exists a universal constant $\beta \geq 1$ such that for all subsets $C \subseteq [n]$, 
\[
	\Ex_I \big \| \P_{\bR_\mi | W_C} \cdot \P_{XY} \cdot \Qsf_{AB | \br_\mi, x,y} - \P_{\bX_i \bY_i \bR_\mi \bA_i \bB_i | W_C} \big \| \,=\, \beta \,\delta^{1/16}/\alpha^3\;,
\]
where $\delta$ is defined in~\eqref{eq:delta} and we identify $(x,y,a,b)$ with $(\bx_i,\by_i,\ba_i,\bb_i)$. 
\end{lemma}

\begin{proof}
We start with two claims, from which the proof of the lemma follows.

\begin{claim}\label{claim:main-1}
For all $\br_\mi, x,y$ and $a,b$,
\begin{equation}\label{eq:anchorpr_main-2}
\Tr \big ( \what{A}_{\br_\mi, x}(a) \otimes \what{B}_{\br_\mi, y}(b) \,\,\,\wt{\Phi}_{\ssxy} \big )  \,=\,\P_{\bA_i \bB_i | \br_\mi,x,y}(a,b)\;.
\end{equation}
\end{claim}

\begin{proof}
From the definitions of $\wt{\Phi}_{\ssxy}$, $\what{A}_{\br_\mi, x}(a)$, and $\what{B}_{\br_\mi, y}(b)$ the left-hand side of~\eqref{eq:anchorpr_main-2} can be expanded as
\begin{align*}
&\Tr \big ( \what{A}_{\br_\mi, x}(a) \otimes \what{B}_{\br_\mi, y}(b) \,\,\,\wt{\Phi}_{\ssxy} \big )	\\
&= \gamma_\ssxy^{-2} \cdot \sum_{\substack{\ba | a, \ba_C \\ \bb | b, \bb_C}}  \Tr \left ( A_{\bomega_\mi,x}(\ba) \otimes B_{\bomega_\mi,y}(\bb) \,\, \ketbra{\psi}{\psi} \right ) \notag \\
	&= \P_{\bA_C \bB_C|\bomega_\mi, x, y}(\ba_C,\bb_C)^{-1} \cdot \sum_{\substack{\ba | a, \ba_C \\ \bb | b, \bb_C}}  \Tr \left ( A_{\bomega_\mi,x}(\ba) \otimes B_{\bomega_\mi,y}(\bb) \,\, \ketbra{\psi}{\psi} \right ) \notag \\
	&= 	\P_{\bA_C \bB_C|\bomega_\mi, x,y}(\ba_C,\bb_C)^{-1} \cdot \Ex_{\substack{\bX| \bomega_\mi, x \\ \bY| \bomega_\mi ,y}} \sum_{\substack{\ba | a, \ba_C \\ \bb | b, \bb_C}}  \Tr \left ( A_{\bx}(\ba) \otimes B_{\by}(\bb) \, \ketbra{\psi}{\psi} \right )\notag \\
	&= 	\P_{\bA_C \bB_C|\bomega_\mi, x, y}(\ba_C,\bb_C)^{-1} \cdot \Ex_{\bX \bY | \bomega_\mi ,x, y} \sum_{\substack{\ba | a, \ba_C \\ \bb | b, \bb_C}}  \Tr \left ( A_{\bx}(\ba) \otimes B_{\by}(\bb) \, \ketbra{\psi}{\psi} \right ) \notag \\	
	&= \P_{\bA_C \bB_C|\bomega_\mi ,x,y}(\ba_C,\bb_C)^{-1} \cdot \P_{\bA_C\bB_C\bA_i \bB_i | \bomega_\mi,x,y}(\ba_C,\bb_C,a,b) \notag \\
	&= \P_{\bA_i \bB_i | \br_\mi,x,y}(a,b)\;,\notag
\end{align*}
where the third line is by~\eqref{eq:gamma-def-1}, the fourth uses the definition of $A_{\bomega_\mi,x}(\ba)$, $B_{\bomega_\mi,y}(\bb)$ in~\eqref{eq:states-operators-2}, and the fifth uses Item 1 of Claim~\ref{claim:ufacts}, which implies that $\bX, \bY$ are independent conditioned on $\bOmega_\mi = \bomega_\mi$, $\bX_i = x$ and $\bY_i = y$.
\end{proof}

\begin{claim}\label{claim:main-2}
The following holds:
\begin{equation}\label{eq:main-2-0}
\Ex_I \big\| \P_{\bR_\mi |W_C} \cdot \P_{\bX_i \bY_i} \cdot \Qsf_{\bA_i \bB_i |\bR_\mi \bX_i \bY_i } 
	-  \P_{\bR_\mi |W_C} \cdot \P_{\bX_i \bY_i} \cdot \P_{\bA_i \bB_i | \bR_\mi \bX_i \bY_i} \big \|  
\,=\,O(\delta^{1/16}/\alpha^3)\;.
\end{equation}
\end{claim}

\begin{proof}
Fix $\br_\mi, x,y$. We bound the total variation distance
\begin{align}
	\big \| \Qsf_{\bA_i \bB_i | \br_\mi, x,y} - \P_{\bA_i \bB_i | \br_\mi,x,y}\big \|
	&\leq  \big \|  \big (   U_{\br_\mi,x} \otimes V_{\br_\mi,y} \big) \wt{\Phi}_{\sspp} \big (   U_{\br_\mi,x} \otimes V_{\br_\mi,y} \big)^\dagger -  \wt{\Phi}_{\br_\mi,x,y}   \big \|_1 \notag\\
	&\leq \sqrt{2} \,\big\| U_{\br_\mi,x} \otimes V_{\br_\mi,y} \ket{\wt{\Phi}_\sspp} - \ket{\wt{\Phi}_{\ssxy}} \big\|\;.\label{eq:main-2-1}
\end{align}
Here the first line follows by contractivity of the trace distance; indeed~\eqref{eq:main-lemma-0} and \eqref{eq:anchorpr_main-2} imply that the distributions $\P$ and $\Qsf$ on the left-hand side can be obtained by measuring the corresponding state on the right-hand side using the POVM  $\{\what{A}_{\br_\mi, x}(a) \otimes \what{B}_{\br_\mi, y}(b)\}_{a,b}$. The second line follows from the fact that for  pure states $\ket{\psi}$ and $\ket{\phi}$, $ \| \psi - \phi \|_1 \leq \sqrt{2} \|\, \ket{\psi} - \ket{\phi} \|$.
Thus
\begin{align*}
\Ex_I \,\, \big\| \P_{\bR_\mi |W_C} \cdot \P_{\bX_i \bY_i} \cdot \Qsf_{\bA_i \bB_i |\bR_\mi \bX_i \bY_i } 
	&-  \P_{\bR_\mi |W_C} \cdot \P_{\bX_i \bY_i} \cdot \P_{\bA_i \bB_i | \bR_\mi \bX_i \bY_i} \big \|  \\
	&= \Ex_I \, \, \Ex_{\bR_\mi |W_C} \,\, \Ex_{\bX_i \bY_i} \,\, \big \| \Qsf_{\bA_i \bB_i |\br_\mi \bx_i \by_i} - \P_{\bA_i \bB_i | \br_\mi \bx_i \by_i} \big \|    \\
	&\leq \sqrt{2} \, \Ex_I \, \, \Ex_{\bR_\mi |W_C} \,\, \Ex_{XY} \big\| U_{\br_\mi,x} \otimes V_{\br_\mi,y} \ket{\wt{\Phi}_\sspp} - \ket{\wt{\Phi}_{\ssxy}} \big\|\\
&\leq O(\delta^{1/16}/\alpha^3)\;,
\end{align*}
where the first inequality is by~\eqref{eq:main-2-1} and the last inequality follows from Proposition~\ref{prop:local_unitaries}.
\end{proof}

Item 1 of \Cref{lem:classical_skew}, combined with the data processing inequality (\Cref{lem:data-processing}) to marginalize over $\bOmega_i$, implies that
\[
	\Ex_I \| \P_{\bX_i \bY_i} - \P_{\bX_i \bY_i | W_C} \| \leq \sqrt{\delta}~.
\]
This allows us to replace the second occurrence of $\P_{\bX_i \bY_i}$ in Item 4 of \Cref{lem:classical_skew} with $\P_{\bX_i \bY_i | W_C}$, only incurring an additive $\sqrt{\delta}$ error:
\begin{equation}
\label{eq:main-3}
	\Ex_I \big\|  \P_{\bX_i \bY_i}  \P_{\bR_\mi | \bX_i = \dummy, \bY_i = \dummy, W_C} -  \P_{\bR_\mi \bX_i \bY_i | W_C} \big\| \leq O(\sqrt{\delta}/\alpha^2)~.
\end{equation}
Using the data processing inequality (\Cref{lem:data-processing}) to marginalize over $\bX_i \bY_i$ gives
\begin{equation}
\label{eq:main-4}
	\Ex_I \big\|  \P_{\bR_\mi | \bX_i = \dummy, \bY_i = \dummy, W_C} -  \P_{\bR_\mi | W_C} \big\| \leq O(\sqrt{\delta}/\alpha^2)~.
\end{equation}
Equation~\eqref{eq:main-4} allows us to replace $\P_{\bR_\mi | \bX_i = \dummy, \bY_i = \dummy, W_C}$ in~\eqref{eq:main-3} with $ \P_{\bR_\mi | W_C}$ while only incurring an additive $O(\sqrt{\delta}/\alpha^2)$ error, yielding
\[
\Ex_I \left \| \P_{\bR_\mi |W_C} \cdot \P_{\bX_i \bY_i} - \P_{\bR_\mi \bX_i \bY_i | W_C} \right \| \leq O(\sqrt{\delta}/\alpha^2)~.
\] 
Using the preceding inequality we can replace $\P_{\bR_\mi | W_C} \cdot \P_{\bX_i \bY_i}$ in~\eqref{eq:main-2-0} with $\P_{\bR_\mi \bX_i \bY_i | W_C}$ while only incurring an additive $O(\sqrt{\delta}/\alpha^2)$ error, yielding
\[
\Ex_I \left \|  \P_{\bR_\mi |W_C} \cdot \P_{\bX_i \bY_i} \cdot \Qsf_{\bA_i \bB_i |\bR_\mi \bX_i \bY_i } 
	-  \P_{\bR_\mi \bX_i \bY_i \bA_i \bB_i |W_C}  \right \| \,=\, O\big(\delta^{1/16}/\alpha^3\big)\;.
\]
Using the fact that $\P_{\bX_i \bY_i}(x,y) = \P_{XY}(x,y)$, this concludes the proof of \Cref{lem:anchorpr_main_lemma}.
\end{proof}

\subsection{Proof of Theorem~\ref{thm:anchorpr_quantum}} 

Let $\eps > 0$ and $p$ satisfy \Cref{eq:p}. First suppose that $\E(G^n,p) = \infty$; then $\E(G^n,p) \geq \E(G,1 - \eps)$ trivially holds. The other case is that $\E(G^n,p) = d$ for some finite $d$. Let $\strategy^n = (\ket{\psi},A,B)$ be a $d$-dimensional strategy for $G^n$ with value $\eval(G^n,\strategy^n) = p$. Assume without loss of generality that $\ket{\psi}$ is a symmetric state as given in Equation~\eqref{eq:psi-sym}.  Our goal is to show the existence of a $d$-dimensional strategy $\strategy$ for $G$ that succeeds with probability at least $1 - \eps$. 
We start with the following proposition.

\begin{proposition}
\label{prop:subset}
	Let $W$ denote the indicator for winning all $n$ coordinates. Suppose that $n \geq \frac{16}{\eps} \log \frac{4}{\eps \cdot \P(W)}$. Then there exists a set $C \subseteq [n]$ of size at most $t = \frac{8}{\eps} \log \frac{4}{\eps \cdot \P(W)}$ such that
	$$
		\Ex_I\, \P (W_i | W_C) \,\geq \, 1 - \eps/2\;,
	$$
	where $\Ex_I$ denotes the expectation over a uniformly random $i$ chosen from $[n] \setminus C$ and $\P(W_i | W_C)$ denotes the probability, using the strategy $\strategy^n$, of winning the $i$-th instance of $G$ conditioned on winning all instances indexed by $C$.
\end{proposition}

\begin{proof}
	Set $\delta = \eps/8$. Let $W_{> 1 - \delta}$ denote the event that the players win more than $(1 - \delta)n$ instances of $G$ using the strategy $\strategy^n$. To show existence of such a set $C$ we will show that $\Ex_C \Ex_I \P(\neg W_i | W_C) \leq \eps/2$, where $C$ is a (multi)set of $t$ independently chosen indices in $[n]$. This implies that there exists a particular set $C$ such that $\Ex_I \P(\neg W_i | W_C) \leq \eps/2$, which concludes the claim.
	
	First we write, for a fixed $C$ and $i \in [n] \setminus C$,
	\begin{align*}
		\P ( \neg W_i | W_C) = \P(\neg W_i | W_C \wedge W_{> 1 - \delta}) \P(W_{> 1 - \delta} | W_C) + \P(\neg W_i | W_C \wedge \neg W_{> 1 - \delta}) \P(\neg W_{> 1 - \delta} | W_C)\;.
	\end{align*}
	Observe that $\Ex_I \P(\neg W_i | W_C \wedge W_{> 1 - \delta})$ is the probability that, conditioned on winning all coordinates indexed by $C$, a randomly selected coordinate $i \in [n] \setminus C$ happens to designate one of the (at most) $\delta n$ instances that were lost. This is at most $\delta n/(n - t) \leq \eps/4$, where we use our assumption on $t$ from the statement of the proposition. Now observe that 
	\begin{align*}
		\Ex_C \P(\neg W_{> 1 - \delta} | W_C) &\leq \Ex_C \frac{\P(W_C | \neg W_{> 1 - \delta})}{\P(W_C)} \\
		&\leq \frac{1}{\P(W)} (1 - \delta)^t \\
		&\leq \eps/4\;,
	\end{align*}
	where in the second inequality we used the fact that $\P(W_C) \geq \P(W)$. Therefore
	\[
		\Ex_C \Ex_I \P ( \neg W_i | W_C) \leq (\eps/4) \cdot 1 + 1 \cdot (\eps/4) \leq \eps/2~,
	\]
	as desired.
\end{proof}

The lower bound on $\P(W)$ given by Eq.~\eqref{eq:p}, for a setting of the universal constant $c$ given in~\eqref{eq:c} below, satisfies the condition of \Cref{prop:subset}: 
\[
	n \geq \frac{n}{s} \geq \frac{1}{c \, \alpha^{48} \, \eps^{17}} \, \ln \frac{4}{\eps \, \P(W)} \geq \frac{16}{\eps} \log \frac{4}{\eps \cdot \P(W)}\;,
\]
where we used that $0 < \alpha, \eps \leq 1$, $s= \max \{ \log |\A \times \B|, 1 \} \geq 1$, and $0 < c \leq \frac{1}{16 \cdot \log(e)}$. Fix a subset $C \subseteq [n]$ satisfying the conclusions of the proposition. It follows that sampling a uniformly random $i \in [n] \setminus C$ and then sampling from the distribution $\P_{\bX_i \bY_i \bR_\mi \bA_i \bB_i | W_C}$ yields a tuple $(i,\bx_i,\by_i,\br_\mi,\ba_i,\bb_i)$ such that $V(\bx_i,\by_i,\ba_i,\bb_i) = 1$ (i.e.\ $W_i = 1$) with probability at least $1 - \eps/2$. \Cref{lem:anchorpr_main_lemma} and the definition of $\Qsf_{AB}$ implies that if we first sample a uniformly random $i$, sample $\br_\mi$ from the distribution $\P_{\bR_\mi | W_C}$ and then play the game $G$ using the strategy $\strategy_{\br_\mi}$ the resulting distribution over tuples $(\br_\mi,x,y,a,b)$ is $O(\delta^{1/16}/\alpha^3)$-close to $\P_{\bX_i \bY_i \bR_\mi \bA_i \bB_i | W_C}$, on average over the choice of index $i$. As a consequence of the two previous points, sampling $i$, sampling from the distribution $\P_{\bR_\mi | W_C}$, and then playing the game $G$ using strategy $\strategy_{\br_\mi}$ yields winning answers with probability at least $1 - \eps/2 - \beta \delta^{1/16}/\alpha^3$, where $\beta$ is the universal constant from \Cref{lem:anchorpr_main_lemma}. 

Using that by assumption $t \leq \frac{8}{\eps} \log \frac{4}{\eps \cdot \P(W)} \leq n/2$,
\begin{align*}
\delta &= \frac{1}{n-t} \Big( \log \frac{1}{\P(W_C)} + t \cdot s \Big)\\
& \leq \frac{2}{n} \Big ( \frac{16 \cdot s}{\eps} \log \frac{4}{\eps \cdot \P(W)} \Big) \\
&\leq \frac{2}{n} \cdot \frac{16 \cdot s}{\eps} \cdot \frac{c \log (e) \, \alpha^{48} \, \eps^{17} \, n}{s} \\
&= 32 \, c \, \log(e) \, \alpha^{48} \eps^{16}~.
\end{align*}
Setting 
\begin{equation}
\label{eq:c}
	c = \frac{1}{32 \, \log(e) \, (4\beta)^{16}}
\end{equation}
we get that $\beta \delta^{1/16}/\alpha^3 \leq \eps/4$, meaning that the probability that the strategy $\strategy_{\br_\mi}$ wins $G$ is at least $1 - \eps$.
By averaging there must exist a pair $(i,\br_\mi)$ such that
\[
	\eval(G,\strategy_{\br_\mi}) \geq 1 - \eps\;.
\]
Since $\strategy_{\br_\mi}$ is a $d$-dimensional strategy where $d = \E(G^n,p)$, this implies that $\E(G,1 - \eps) \leq d$, which concludes the proof of \Cref{thm:anchorpr_quantum}.

\bibliography{anchor}

\end{document}